\documentclass[journal, twoside]{IEEEtran}
\usepackage{mathrsfs}  

\usepackage{epsfig} 
\usepackage{times} 
\usepackage{amsmath} 
\usepackage{amssymb}  
\usepackage{algorithmic}
\usepackage{subfigure}
\usepackage{color}
\usepackage{cite}

\usepackage{algorithm}
\usepackage{bm}
\usepackage{graphicx}
\usepackage{multirow}
\usepackage{array}
\usepackage{graphics} 

\allowdisplaybreaks

\floatname{algorithm}{\emph{Algorithm}}

\newcommand{\cS}{\mathcal{S}}

\def\bull{\vrule height 1.1ex width 1.1ex depth -.0ex }

\newtheorem{theorem}{Theorem}

\newtheorem{proposition}{Proposition}
\newtheorem{corollary}{Corollary}
\newtheorem{definition}{Definition}

\title{Deterministic and Nonblocking Supervisory Control of Discrete Event Systems under Cyber Attacks}

\author{Feng Lin,~\IEEEmembership{Fellow,~IEEE}, Caisheng Wang,~\IEEEmembership{Fellow,~IEEE}, Jun Chen,~\IEEEmembership{Senior Member,~IEEE}, and Xiang Yin,~\IEEEmembership{Senior Member,~IEEE} 
	\thanks{The authors are supported in part by the National Science Foundation of USA under grants 2146615, 2432098,  and 2432099.}
	\thanks{Feng Lin and Caisheng Wang are with the Department of Electrical and Computer Engineering, Wayne State University, Detroit, MI 48202, USA (e-mail: flin@wayne.edu and cwang@wayne.edu). }
	\thanks{Jun Chen is with the Department of Electrical and Computer Engineering, Oakland University, Rochester, MI 48374, USA (e-mail: junchen@oakland.edu). }
	\thanks{Xiang Yin is with the Department of Automation, Key Laboratory of System Control and Information Processing, Shanghai Jiao Tong University, Shanghai 200240, China (e-mail:yinxiang@sjtu.edu.cn). }
}

\begin{document}

\maketitle

\begin{abstract}
	
	We investigate deterministic and nonblocking supervisory control of discrete event systems under cyber-attacks using the ALTER (Attack Language for Transition-basEd Replacement) model. 
	While prior works consider supervisory control that achieves either the large (upper bound) language or small (lower bound) language separately, deterministic supervisory control achieves both large language and small language at the same time to ensure that the language generated by the supervised system is unique and deterministic. 
	We introduce two new concepts of CA-D-controllability and CA-D-observability and prove that they are necessary and sufficient for the existence of a deterministic supervisor. 
	For nonblocking supervisory control, the objective is to ensure that the supervised system can always reach marked states under any attack scenario. We prove that relative closure, CA-D-controllability, and CA-D-observability together are necessary and sufficient for the existence of a nonblocking supervisor. 
	We further develop methods to verify CA-D-controllability and CA-D-observability. We also illustrate our results using a robotic system example.

	
\end{abstract}


\section{Introduction} \label{s1}

The modern world is increasingly governed by Cyber-Physical Systems (CPS) with complex integrations of cyber entities, such as controllers, software, agents, and embedded systems, with physical processes. 
These systems feature a closed feedback loop, where sensors measure the physical world (plant); cyber components (controllers) process the data and make control decisions; and actuators execute commands to control the physical world. 
From smart grids and autonomous vehicles to industrial robots and medical devices, CPS are the backbone of critical infrastructure, manufacturing, and many other engineering systems, driving unprecedented efficiency and innovation.

However, this fusion of the digital and physical realms creates a new dimension of vulnerability. 
Cyber-attacks on CPS transcend traditional data breaches; they are attacks designed to cause deliberate physical disruption, damage, and even harm to human life. 
Unlike conventional IT systems where the primary targets of cyber-attacks are data and confidentiality, the primary targets of cyber-attacks in a CPS are safety, reliability, and functionality. 
Attackers seek to sabotage this critical cyber-physical feedback loop by manipulating sensor data to hide failures, injecting malicious commands to destroy equipment, or shutting down systems entirely to cause chaos.

Understanding and mitigating the threat of cyber-attacks are paramount. As our dependence on these intelligent systems grows, so does the imperative to secure them against cyber-attacks that threaten not only our data, but also our physical well-being and societal stability.

A CPS can often be modeled as a discrete event system (DES) that is characterized by a discrete state space and event-driven dynamics, or as a hybrid system that consists of both discrete events and continuous variables. 
While DES has been well investigated since the 1980s, the investigation of cyber-attacks in DES is relatively new. We review some key results on cyber-attacks in DES below. 
They are not exhaustive. For additional results on cyber-attacks in DES, the reader is referred to the surveys and monographs such as \cite{rashidinejad2019supervisory, hadjicostis2022cybersecurity, su2024cybersecurity}.

In \cite{su2018supervisor}, the protection of DES against sensor reading alteration attacks is studied, where adversaries manipulate sensor data to induce harmful supervisor actions. 
The attack-with-bounded-sensor-reading-alterations (ABSRA) problem is formulated, and it is shown that the supremal ABSRA exists and is computable for regular plant and supervisor models. 
A synthesis algorithm for ABSRA-robust supervisors is provided, ensuring that attacks are either detectable or harmless.

In \cite{fritz2018modeling}, a detection method for cyber-attacks in DES-modeled CPS is proposed, focusing on deception attacks such as replay and covert attacks. 
By employing permutation matrices to disrupt transmitted signals, the method enables effective detection and mitigation. The applicability of the method is demonstrated on a laboratory manufacturing system.

\cite{zhang2018framework} addresses the safety supervisory control problem under both sensor and actuator attacks. A special automaton, called the attack structure, is defined as the parallel composition of two specific structures. 
This automaton can be used either by attackers to guide the plant toward unsafe states while remaining hidden from the supervisor, or by supervisors to evaluate robustness against such attacks.

In \cite{wang2021supervisory}, a product observation reachability graph is introduced to determine the existence of resilient supervisors under actuator and sensor attacks modeled via labeled Petri nets. 
The same graph is further utilized to synthesize supervisors that detect attacks and disable controllable transitions before unsafe states are reached.

\cite{wakaiki2019supervisory} considers supervisory control of DES with multiple adversaries. It is shown that a resilient supervisor exists if the desired language is both controllable and observable with respect to adversaries. 
Standard DES supervisory control tools can then be applied to test the existence and construct the supervisor.

\cite{alves2022discrete} studies DES under partial observation in the presence of multiple attackers capable of altering outputs. 
The concept of P-observability is introduced to guarantee that supervisors can distinguish between attacked and unattacked outputs requiring different control actions. 
Definitions, visual representations, and algorithms are provided to verify P-observability for legal languages represented as automata.

In \cite{chen2022attackable}, detectability under sensor attacks is analyzed. The notions of attackable strong and weak detectability are introduced. Necessary and sufficient conditions are derived to assess system resilience.

\cite{su2023decidability} examines the cybersecurity of DES against smart sensor attacks. A condition based on risky pairs (damage strings and observable strings) is formulated to determine the existence of such attacks. 
A novel encoding scheme eliminates risky pairs, leading to the decidability of resilient nonblocking supervisor existence. A synthesis procedure is also provided.

\cite{wang2023attack} develops a finite-state transducer model for attacks such as insertion and deletion. A polynomial-time algorithm is proposed for resilient supervisor design, and an open-source tool is released for practical applications.

\cite{alves2024persistent} introduces persistent attackers in DES, which can stealthily modify communication between controllers and subsystems. 
A new attack model and corresponding design methods are developed, enabling such attackers to remain undetected while compromising system behavior.

In \cite{tai2023covert}, the synthesis of covert sensor-actuator attacks is investigated. A transformation-based approach converts the covert attack synthesis problem into a supervisory control problem. 
The work is extended to distributed DES in \cite{tai2023synthesis}, where a heuristic-based incremental synthesis method with finite-state automata is developed.


Our work also examines DES under cyber-attacks, presenting an approach that both differs from and complements existing research. We introduce a novel sensor attack model named the ALTER (Attack Language for Transition-basEd Replacement) model \cite{zheng2021modeling, zheng2024modeling, lin2023diagnosability, lin2024diagnosability}. In this model, the event of an attackable transition can be replaced by any string from its associated attack language. A key characteristic of the ALTER model is its dual nature: it is general enough to encapsulate common sensor attacks like deletions, insertions, replacements, and all-out attacks, while also being specific by explicitly defining each attackable transition with its corresponding attack language. This combination makes the ALTER model well-suited for the analysis and synthesis of robust supervisors in the face of cyber-attacks.

The ALTER model has been applied in supervisory control, diagnosability, detectability, and opacity analysis of DES \cite{zheng2021modeling, zheng2024modeling, 
wang2024coordination,  wang2025supervisory, lin2023diagnosability, lin2024diagnosability, ritsuka2025detectability}. 
In supervisory control, it is shown in \cite{zheng2024modeling} that cyber-attacks cause the supervised system to generate nondeterministic languages that are not unique.

To ensure the supervised system is safe and legal, the upper bound of all possible languages that may be generated under cyber-attacks, called the \emph{large language}, must be contained in a given legal language $J_a$. 
To this end, two concepts, CA-controllability and CA-observability (CA for cyber-attacks), are introduced. It is proven in \cite{zheng2024modeling} that there exists a supervisor whose large language equals $J_a$ if and only if $J_a$ is both CA-controllable and CA-observable.

To ensure the supervised system performs the minimally required tasks, the lower bound of all possible languages that may be generated under cyber-attacks, called the \emph{small language}, must  contain a given required language $J_r$. 
To do so, two concepts, CA-S-controllability and CA-S-observability (S for the small language) are introduced. It is proven in \cite{wang2025supervisory} that there exists a supervisor whose small language equals $J_r$ if and only if $J_r$ is both CA-S-controllable and CA-S-observable.


This paper studies deterministic and nonblocking supervisory control under cyber-attacks. A supervisor is defined as deterministic when its large language and small language are identical.
Evidently, if the upper and lower bounds of all possible languages that can be generated under cyber-attacks are the same, then the language produced by the supervised system is deterministic and unique.
To establish the condition for the existence of deterministic control, we present two new concepts: CA-D-controllability and CA-D-observability (with D denoting deterministic).
We demonstrate that a supervisor exists for which the large language and small language are both equal to $J$ if and only if $J$ is both CA-D-controllable and CA-D-observable.
Methods to verify CA-D-controllability and CA-D-observability are also developed.

{Note that the existence condition for deterministic control cannot be obtained by combining the existence conditions for large-language control and small-language control. This is because, even if the existence conditions for large-language control and small-language control are both satisfied, there is no guarantee that the supervisor achieving the large language is the same as the supervisor achieving the small language. }

{A supervisor is called nonblocking if the supervised system can always reach some marked (final) states and in the process, no event enabled by the supervisor under one cyber-attack can be disabled by the supervisor under another cyber-attack.}
In other words, reaching the marked states is guaranteed under all possible cyber-attacks. We prove that there exists a nonblocking supervisor whose marked language is equal to $K$ if and only if $K$ is relatively closed and its prefix closure is both CA-D-controllable and CA-D-observable. 

The principal contributions of this work are fourfold. First, we propose the deterministic and nonblocking supervisory control problems to address the control requirements of cyber-physical systems subjected to cyber-attacks. Second, we introduce the concepts of CA-D-controllability and CA-D-observability and prove that they provide the existence condition for a deterministic supervisory control solution. Third, we derive the existence condition for nonblocking supervisory control under cyber-attacks. Fourth, we develop a verification method for CA-D-controllability and CA-D-observability, implementing it using automata. 

The paper is organized as follows. In Section II, discrete event systems and supervisory control under cyber-attacks, including the ALTER model, are reviewed. 
In Section III, deterministic supervisory control is investigated. CA-D-controllability and CA-D-observability are introduced, and the existence condition for deterministic supervisor is derived. 
Section IV considers nonblocking supervisory control. Marked states and nonblocking supervisor are defined. Existence condition for nonblocking supervisor is derived. 
In Section V, methods to check CA-D-controllability and CA-D-observability are developed. Section VI presents an example of a robotic system to illustrate the results of the paper.

\section{Discrete event systems under cyber-attacks} \label{s2}


This section provides a brief overview of supervisory control for discrete-event systems under cyber-attacks, as established in prior work \cite{zheng2021modeling, zheng2024modeling}. The system to be controlled, known as the plant, is modeled as a deterministic automaton \cite{cassandras2021introduction}
$$
G=( X, \Sigma , \xi , x_0, X_m ),
$$
where $X$ is the state set, $\Sigma$ is the event set, $\xi : X \times \Sigma \rightarrow X$ is the partial transition function, $x_0$ is the initial state, and $X_m$ is the set of marked states.
We assume, without loss of generality, that the automaton $G$ is trim, meaning every state is both reachable from the initial state $x_0$ and can reach a marked state in $X_m$. The set of all possible transitions is also represented by $\xi$, defined as $\xi = \{ (x, \sigma, x'): \xi (x, \sigma) =x' \}$.



The set of all strings formed from the event set $\Sigma$ is denoted by $\Sigma ^*$. The transition function $\xi$ is extended to the domain of strings, $\xi : X \times \Sigma^* \rightarrow X$, following the standard convention. The notation $\xi (x,w)!$ signifies that the transition $\xi (x,w)$ is defined. The language \emph{generated} by $G$ is defined as the set of all strings that are defined in $G$ starting from the initial state $x_0$, which is:
$$
L(G) = \{ w \in \Sigma ^*: \xi (x_0,w)! \}.
$$
The language {\em marked} by $G$ is the set of all strings defined in $G$ from the initial state $x_0$ to a marked state, that is,
$$
L_m(G) = \{ w \in L(G): \xi (x_0,w) \in X_m \}.
$$


A language $L \subseteq \Sigma ^*$ is, in general, a set of strings. For a given string $w \in \Sigma ^*$, the notation $w' \leq w$ indicates that $w'$ is a prefix of $w$. The (prefix) closure of $L$, written as $\overline{L}$, is the collection of all prefixes of all strings in $L$, specifically:
$$
\overline{L}=\{w' \in \Sigma^* :(\exists w\in L)w' \leq w\}
$$
A language is (prefix) closed if it equals its prefix closure. By the definition, $L(G)$ is closed. Since $G$ is trim, we have 
$$
\overline{L_m(G)} = L(G).
$$


A supervisor is used as a controller to achieve a specific control objective by governing the system's behavior. This supervisor possesses the ability to control a subset of events and observe another subset. The events that can be controlled, known as \emph{controllable events}, are denoted by $\Sigma _c$ ($\subseteq \Sigma$). The remaining events, which are uncontrollable, form the set $\Sigma _{uc} = \Sigma - \Sigma _c$. Similarly, the events that can be observed, called \emph{observable events}, are denoted by $\Sigma _o$ ($\subseteq \Sigma$). The complement of this set constitutes the unobservable events, $\Sigma _{uo} = \Sigma - \Sigma _o$.

For a given string, its observation is described by the \emph{projection} $P:\Sigma^* \rightarrow \Sigma^*_o$ which is defined as
\begin{equation*}
	\begin{aligned}
		& P(\varepsilon)= \varepsilon \\
		& P(\sigma) =\begin{cases}
			\sigma &\text{if $\sigma \in \Sigma_o$} \\
			\varepsilon &\text{if $\sigma\in \Sigma _{uo}$} \\
		\end{cases}\\
		& P(w\sigma)=P(w)P(\sigma) ,  w\in \Sigma^*, \sigma \in \Sigma,
	\end{aligned}
\end{equation*}
where $\varepsilon$ is the empty string.

After the occurrence of $w\in L(G)$, a \emph{supervisor} $\cS$ observes $v \in P (w)$. Based on $v$, $\cS$ enables a set of events, denoted by $\cS(v)$. Hence, $\cS$ is a map 
$$
\cS:P(L(G)) \to \Gamma,
$$
where $\Gamma = 2^{\Sigma}$ is the set of all possible controls. It is required that $\Sigma_{uc}\subseteq \cS(v)$ because uncontrollable events cannot be disabled. 

The {\em closed-loop system\/} is denoted by $\cS/G$. The language generated by $\cS/G$, denoted by $L(\cS/G)$, is defined recursively as 
\begin{enumerate}
	\item
	The empty string belongs to $L(\cS/G)$, that is,
	$$
	\varepsilon \in L(\cS/G).
	$$

	\item
	If $w$ belongs to $L(\cS/G)$, then for all $\sigma \in \Sigma$, $w \sigma$ belongs to $L(\cS/G)$ if and only if $w \sigma$ is physically possible in $L(G)$ 
	and $\sigma$ is uncontrollable or enabled by $\cS$ after observing $v=P(w)$, that is, for all $ w \in L(\cS/G)$ and $\sigma \in \Sigma$,
	\begin{align*}
		&  w \sigma\in L(\cS/G) \\
		\Leftrightarrow \
		& w \sigma \in L(G) \wedge (\sigma \in \Sigma _{uc} \vee \sigma \in \cS(P(w))). 
	\end{align*}
\end{enumerate}


In conventional supervisory control with partial observations, the objective for a given closed specification language $J \subseteq L(G)$ is to design a supervisor $\cS$ that satisfies $L(\cS/G)=J$. The necessary and sufficient condition for the existence of such a supervisor is defined by the properties of controllability and observability. A closed language $J$ is said to be \emph{controllable} with respect to $L(G)$ and $\Sigma_{uc}$ if
$$
{J} \Sigma_{uc} \cap L(G) \subseteq {J}.
$$
$J$ is {\em observable} with respect $L(G)$ and $\Sigma_o$ \cite{lin1988observability} if 
\begin{align*}
	& (\forall w \in \Sigma ^*)(\forall \sigma \in \Sigma) w \sigma \in J \\
	& \Rightarrow (\forall w' \in P^{-1} (P(w))) (w' \in J \wedge w' \sigma \in L(G) \Rightarrow w' \sigma \in J) .
\end{align*}


As established in \cite{lin1988observability}, a supervisor $\mathcal{S}$ satisfying $L(\mathcal{S}/G)=J$ exists if and only if the language $J$ is controllable with respect to $L(G)$ and $\Sigma_{uc}$ and observable with respect to $L(G)$ and $\Sigma_o$.


{Cyber-attackers can alter a supervisor's observations by targeting the sensors or the communication channel from the plant to the supervisor.} Consequently, we operate under the assumption that certain observable events are vulnerable to attack. To describe these sensor attacks, we adopt the \emph{ALTER} (Attack Language, Transition-basEd, Replacement) attack model from prior work \cite{zheng2021modeling, zheng2024modeling, lin2023diagnosability, lin2024diagnosability}, which is reviewed as follows (for a comprehensive explanation, see \cite{zheng2024modeling, lin2024diagnosability}).


The set of observable events that are vulnerable to attack, referred to as \emph{attackable observable events}, is denoted by $\Sigma _o^a \subseteq \Sigma _o$. The corresponding set of attackable observable transitions is $\xi ^a = \{ (x, \sigma, x') \in \xi : \sigma \in \Sigma _o^a \}$.

For every attackable transition $tr=(x, \sigma, x')\in \xi ^a$, it is assumed that an attacker can replace the observation of event $\sigma$ with any string from a specific attack language $A_{tr} \subseteq \Sigma_o ^*$. This language $A_{tr}$ is determined based on the knowledge of the attacker and the sensor. The collection of all these \emph{attack languages} is represented by $\mathbb{A}$.
$$
\mathbb{A}=\{ A_{tr} \subseteq \Sigma_o ^* :tr=(x, \sigma, x')\in \xi ^a\}.
$$

Sensor attacks are then modeled by a mapping from the set of attackable observable transitions to the set of languages $\mathbb{A}$ 
$$
\pi : \xi ^a\rightarrow \mathbb{A}
$$
where $\pi(tr)=A_{tr}$. This ALTER model is general and includes the following special cases. (1) Deletion: if $A_{\sigma} = \{\varepsilon, ... \}$, then $\sigma$ may be deleted. 
(2) Replacement: if $A_{\sigma} = \{\alpha, ... \}$, then $\sigma$ may be replaced by $\alpha$. 
(3) Insertion: if $A_{\sigma} = \{\sigma \alpha, \alpha \sigma, ... \}$, then $\alpha$ may be inserted. 
(4) All-out attack: if $A_{\sigma} = \Sigma_o ^*$, then the attacker can change $\sigma$ to any observable strings. 

{Note that knowing the attack languages is not equivalent to knowing the attacker. If the attacker is not known at all, then we can simply let $A_{tr} = \Sigma_o ^*$ (all-out attack). 
However, if we do have some knowledge of the attacker, for example, if we know that the attacker is stealthy, then the attack languages are much smaller, because, in order to preserve the stealthiness, the attacker cannot arbitrarily alter the event observations. 
Note also that the larger the attack languages $A_{tr}$ are, the less is known about the attacker. As we will show in the paper, the larger  $A_{tr}$ are, the less it is likely that the existence condition for a resilient supervisor will be satisfied
If the existence condition is not satisfied, then there are two ways to proceed: (1) To know more about the attacker so that the uncertainty described by $A_{tr}$ can be reduced. 
(2) To restrict the behavior of the supervised system to ensure safety/legality of the supervised system. Although it is better that the attack languages $A_{tr}$ are smaller, we do not put any requirements/restrictions on $A_{tr}$ or on how $A_{tr}$ are obtained. A separate but important problem is how to synthesize $A_{tr}$ from the attacker’s point of view. The synthesis may focus on issues such as how to ensure the attacks are stealthy or how to compromise the safety of a supervised system.}

For a string $w \in \Sigma ^*$, we use $|w|$ to denote its length. If a string $w = \sigma _1 \sigma _2 ..., \sigma _{|w|}$ occurs in $G$, the set of possible strings after cyber-attacks in the observation channel, 
called \emph{observation attacks}, is denoted by $\Theta ^\pi(w)$ and obtained as follows. Denote $x_k = \xi(x_0,\sigma _1\cdots\sigma _{k}), k= 1, 2, ..., |w|$, then
$$
\Theta ^\pi(w) = L_1 L_2 ... L_{|w|},
$$
where
\begin{align*}
	L_k = \left\{ \begin{array}{ll}
		\{ \sigma _k \} & \mbox{if } (x_{k-1}, \sigma _k, x_k) \not\in \xi ^a \\
		A_{(x_{k-1}, \sigma _k, x_k)} & \mbox{if } (x_{k-1}, \sigma _k, x_k) \in \xi ^a \end{array} \right .
\end{align*}

Note that $\Theta ^\pi(w)$ may contain more than one string. Hence, $\Theta ^\pi$ is a mapping from $L(G)$ to $2^{\Sigma ^*}$:
$$
\Theta ^\pi: L(G) \rightarrow 2^{\Sigma ^*}.
$$

The observation under both partial observation and observation attacks is then given by
$$
\Phi ^\pi = P \circ \Theta ^\pi,
$$
where $\circ$ denotes composition (of functions). In other words, for $w \in L(G)$, $\Phi ^\pi (w) = P ( \Theta ^\pi (w))$. Hence, $\Phi ^\pi$ is a mapping from $L(G)$ to $2^{\Sigma _o^*}$:
$$
\Phi ^\pi: L(G) \rightarrow 2^{\Sigma _o ^*}.
$$

We extend $P$, $\Theta ^\pi$, and $\Phi ^\pi$ from strings $w$ to languages $L \subseteq L(G)$ in the usual way as
\begin{align*}
	& P(L) = \{ v \in \Sigma^*_o: (\exists w \in L) v = P(w) \} \\
	& \Theta ^\pi(L) = \{ v \in \Sigma^*: (\exists w \in L) v  \in
	\Theta ^\pi(w) \} \\
	& \Phi ^\pi(L) = \{ v \in \Sigma^*_o: (\exists w \in L) v  \in \Phi ^\pi(w) \}.
\end{align*}


The ALTER model possesses a dual nature, being both general and specific. Its generality stems from its capacity to encompass most attack types such as deletion, insertion, replacement, and all-out attacks. Its specificity lies in the precise definition of attacks through individual attack languages. This novel approach enables us to address numerous practical problems and derive concrete solutions.

Under cyber-attacks, a supervisor $\cS$ is now a mapping
$$
\cS :\Phi ^\pi (L(G)) \rightarrow \Gamma .
$$


We further assume that attackers can compromise the control channel from the supervisor to the plant, altering the disablement or enablement of certain controllable events. This means an attacker can enable an event that the supervisor has disabled, or disable an event that the supervisor has enabled.
The set of controllable events vulnerable to such attacks, termed \emph{attackable controllable events}, is denoted by $\Sigma _c^a \subseteq \Sigma _c$. It is important to note that uncontrollable events are always permitted to occur and are not subject to disablement by attackers.


In the presence of cyber-attacks on the control channel, referred to as \emph{control attacks}, the intended control command $\gamma \in \Gamma$ can be altered. An attacker can add events from $\Sigma _c^a$ to $\gamma$ or remove events from $\gamma$ that are in $\Sigma _c^a$. Consequently, the set of possible control commands after such attacks is:
\begin{align*}
	\Delta(\gamma) = \{ \gamma_a \in \Gamma :  (\exists \gamma ', \gamma '' \subseteq  \Sigma _c^a ) \gamma_a = (\gamma - \gamma ')\cup  \gamma ''\} .
\end{align*}


When a supervisor issues a control command $\mathcal{S}(v)$ based on an observation $v \in \Phi ^\pi (L(G))$, this command is subject to alteration by control attacks. The set of all possible control commands that the plant may actually receive under such attacks is denoted by $\mathcal{S}^a (v)$, is defined as
\begin{align*}
	\cS^a(v)=\Delta(\cS(v)).
\end{align*}

%

The supervised system under cyber-attacks is denoted as $\cS ^a/G$. Due to these attacks, the language generated by this system, $L(\cS ^a/G)$, becomes nondeterministic and is therefore not guaranteed to be unique.

To ensure the legality and safety of the supervised system, we consider the upper bound of all possible languages that $\cS ^a/G$ can generate. This upper bound is referred to as the \emph{large language}, defined as follows.

\begin{definition} 
	The {large language} of the supervised system $\cS ^a/G$, denoted by $L_a (\cS ^a/G)$, is defined recursively as follows.
	\begin{enumerate}
		\item
		The empty string belongs to $L_a (\cS ^a/G)$, that is,
		$$
		\varepsilon \in L_a (\cS ^a/G).
		$$
		\item
		If $w$ belongs to $L_a (\cS ^a/G)$, then for all $\sigma \in \Sigma$, $w \sigma$ belongs to $L_a (\cS ^a/G)$ if and only if $w \sigma$ is physically possible in $L(G)$ 
		and $\sigma$ is uncontrollable or enabled by $\cS ^a$ in {\em some} situations, that is, for all $ w \in L_a (\cS ^a/G)$ and $\sigma \in \Sigma$,
		\begin{equation} \label{LargeL}
			\begin{split}
				&  w \sigma\in L_a (\cS ^a/G) \\
				\Leftrightarrow \
				& w \sigma \in L(G) \wedge (\sigma \in \Sigma _{uc} \\
				& \vee (\exists v \in \Phi ^\pi (w))(\exists \gamma \in \cS ^a(v)) \sigma \in \gamma). 
			\end{split}
		\end{equation}
	\end{enumerate}
\end{definition} 

\vspace{0.1in}


To fulfill certain minimum performance requirements, we also consider the lower bound of all potential languages generated by the supervised system $\cS ^a/G$. This lower bound is termed the \emph{small language} and is defined as follows.

\begin{definition}
	The {small language} of the supervised system $\cS ^a/G$, denoted by $L_r (\cS ^a/G)$, is defined recursively as follows.
	\begin{enumerate}
		\item
		The empty string belongs to $L_r (\cS ^a/G)$, that is,
		$$
		\varepsilon \in L_r (\cS ^a/G).
		$$
		\item
		If $w$ belongs to $L_r (\cS ^a/G)$, then for all $\sigma \in \Sigma$, $w \sigma$ belongs to $L_r (\cS ^a/G)$ if and only if $w \sigma$ is physically possible in $L(G)$ 
		and $\sigma$ is uncontrollable or enabled by $\cS ^a$ in {\em all} situations, that is, for all $ w \in L_r (\cS ^a/G)$ and $\sigma \in \Sigma$,
		\begin{equation} \label{SmallL}
			\begin{split}
				&  w \sigma\in L_r (\cS ^a/G) \\
				\Leftrightarrow \
				& w \sigma \in L(G) \wedge (\sigma \in \Sigma _{uc} \\
				& \vee (\forall v \in \Phi ^\pi (w))(\forall \gamma \in \cS ^a(v)) \sigma \in \gamma).
			\end{split}
		\end{equation}
	\end{enumerate}
\end{definition}

\vspace{0.1in}

%

The first control objective is to design a supervisor $\cS :\Phi ^\pi (L(G)) \rightarrow \Gamma$ such that the large language satisfies $L_a (\cS ^a/G)=J_a$. This objective is needed when the goal is to guarantee that all possible behaviors generated by the supervised system under attack remain within a legal, safe, or admissible language $J_a$.

To establish the necessary and sufficient conditions for the existence of such a supervisor, the following two definitions are introduced in \cite{zheng2024modeling}.

\begin{definition} 
	A closed language $J \subseteq L(G)$ is {\em CA-controllable} with respect to $L(G)$, $\Sigma _{uc}$, and $\Sigma ^a_c$ if
	\begin{align*}
		J (\Sigma _{uc} \cup \Sigma ^a_c) \cap L(G) \subseteq J.
	\end{align*}
\end{definition}

\begin{definition} 
	A closed language $J \subseteq L(G)$ is {\em CA-observable} with respect to $L(G)$, $\Sigma _{o}$, $\Sigma ^a_o$ and $\Phi^\pi$ if
	\begin{equation} \label{CAO} 
		\begin{split}
			& (\forall w \in \Sigma ^*)(\forall \sigma \in \Sigma) (w \sigma \in J \\
			& \Rightarrow (\exists v \in \Phi ^\pi (w))(\forall w' \in (\Phi ^\pi)^{-1} (v)) \\
			& (w' \in J \wedge w' \sigma \in L(G) \Rightarrow w' \sigma \in J)),
		\end{split}
	\end{equation}
	where $(\Phi ^\pi)^{-1}$ is the inverse mapping of $\Phi ^\pi$: for $v \in \Phi ^\pi (L(G))$,\\
	\begin{align*}
		(\Phi ^\pi)^{-1} (v) = \{ w' \in L(G) : v \in \Phi ^\pi (w') \}.
	\end{align*}
\end{definition}

\vspace{0.05in}


As established in \cite{zheng2024modeling}, a supervisor $\cS$ satisfying $L_a(\cS^a/G)=J_a$ exists if and only if the language $J_a$ is both CA-controllable with respect to $L(G)$, $\Sigma _{uc}$, and $\Sigma ^a_c$, and CA-observable with respect to $L(G)$, $\Sigma _{o}$, $\Sigma ^a_o$, and $\Phi^\pi$.

The second control objective is to design a supervisor $\cS :\Phi ^\pi (L(G)) \rightarrow \Gamma$ such that the small language $L_r (\cS ^a/G)=J_r$. This is needed when we want to ensure that the languages generated by $\cS$ always contain some minimum required language $J_r$.

To obtain the existence conditions for a supervisor that achieves the small language, the following definitions are given in \cite{wang2025supervisory}.

\begin{definition} 
	A closed language $J \subseteq L(G)$ is {\em CA-S-controllable} with respect to $L(G)$, $\Sigma _{uc}$, and $\Sigma ^a_c$ if
	\begin{align*}
		(J \Sigma_{uc} \cap L(G) \subseteq J) \wedge (J \subseteq (\Sigma - \Sigma^a_c)^*).
	\end{align*}
\end{definition}

\begin{definition} 
	
	A closed language $J \subseteq L(G)$ is {\em CA-S-observable} with respect to $L(G)$, $\Sigma _{o}$, $\Sigma ^a_o$ and $\Phi^\pi$ if
	\begin{align*}
		& (\forall w \in J)(\forall \sigma \in \Sigma) ((w\sigma \in L(G) \wedge (\forall v \in \Phi ^\pi(w)) \\
		&  (\exists w' \in J) v \in  \Phi ^\pi(w') \wedge w'\sigma \in J) \Rightarrow w\sigma \in J). 
	\end{align*}
\end{definition}

\vspace{0.1in}


As demonstrated in \cite{wang2025supervisory}, a supervisor $\cS$ for which $L_r(\cS^a/G)=J_r$ exists if and only if $J_r$ is CA-S-controllable with respect to $L(G)$, $\Sigma _{uc}$, and $\Sigma ^a_c$, and is CA-S-observable with respect to $L(G)$, $\Sigma _{o}$, $\Sigma ^a_o$, and $\Phi^\pi$.

While the control to achieve the large language and the control to achieve the small language are useful in different applications, in some applications, we want the control to achieve both the large language and the small language at the same time. 
In other words, the control objective is to design a supervisor $\cS :\Phi ^\pi (L(G)) \rightarrow \Gamma$ such that $L_a (\cS ^a/G)= L_r (\cS ^a/G)= J$. 
We call such a control \emph{deterministic control}, because $L_a (\cS ^a/G)= L_r (\cS ^a/G)$ means that the language generated by $\cS ^a/G$ is deterministic (that is, unique).
As to be shown, the deterministic control is especially important for nonblocking control.

\section{Deterministic supervisory control} \label{s3}

In this section, we investigate deterministic supervisory control to achieve $L_a (\cS ^a/G)= L_r (\cS ^a/G)= J$. To obtain the existence conditions for a deterministic supervisor, let us introduce two new definitions.

\begin{definition} 
	A closed language $J \subseteq L(G)$ is {\em CA-D-controllable} with respect to $L(G)$, $\Sigma _{uc}$, and $\Sigma ^a_c$ if
	\begin{align*}
		& (J (\Sigma _{uc} \cup \Sigma ^a_c) \cap L(G) \subseteq J) \wedge (J \subseteq (\Sigma - \Sigma^a_c)^*).
	\end{align*}
\end{definition}

\begin{definition} 
	A closed language $J \subseteq L(G)$ is {\em CA-D-observable} with respect to $L(G)$, $\Sigma _{o}$, $\Sigma ^a_o$ and $\Phi^\pi$ if
	\begin{align*}
		& (\forall w, w' \in  J)(\forall \sigma \in \Sigma) (w\sigma \in J \wedge w'\sigma \in L(G) \\
		& \wedge w'\sigma \notin J \Rightarrow \Phi ^\pi(w) \cap \Phi ^\pi(w')=\emptyset ).
	\end{align*}
\end{definition}

\vspace{0.07in}

The following proposition describes the relations among CA-controllability, CA-S-controllability, and CA-D-controllability.

\begin{proposition} \label{proposition2}
	
	$J$ is CA-controllable and CA-S-controllable with respect to $L(G)$, $\Sigma _{uc}$, and $\Sigma ^a_c$ if and only if it is CA-D-controllable with respect to $L(G)$, $\Sigma _{uc}$, and $\Sigma ^a_c$.
	
\end{proposition}
\noindent {\em Proof}

$J$ is CA-controllable and CA-S-controllable with respect to $L(G)$, $\Sigma _{uc}$, and $\Sigma ^a_c$ if and only if
\begin{align*}
	& (J (\Sigma _{uc} \cup \Sigma ^a_c) \cap L(G) \subseteq J) \\
	& \wedge (J \Sigma_{uc} \cap L(G) \subseteq J) \wedge (J \subseteq (\Sigma - \Sigma^a_c)^*) \\
	\Leftrightarrow \
	& (J (\Sigma _{uc} \cup \Sigma ^a_c) \cap L(G) \subseteq J) \wedge (J \subseteq (\Sigma - \Sigma^a_c)^*) \\
	& (\mbox{because } J (\Sigma _{uc} \cup \Sigma ^a_c) \cap L(G) \subseteq J \\
	& \Rightarrow J \Sigma_{uc} \cap L(G) \subseteq J) ,
\end{align*}
that is, $J$ is CA-D-controllable with respect to $L(G)$, $\Sigma _{uc}$, and $\Sigma ^a_c$.

\hfill \bull

The following proposition describes the relations among CA-observability, CA-S-observability, and CA-D-observability.

\begin{proposition} \label{proposition1}
	
	If a closed language $J \subseteq L(G)$ is  CA-D-observable with respect to $L(G)$, $\Sigma _{o}$, $\Sigma ^a_o$ and $\Phi^\pi$, then it is CA-observable and CA-S-observable with respect to $L(G)$, $\Sigma _{o}$, $\Sigma ^a_o$ and $\Phi^\pi$.
	
\end{proposition}
\noindent {\em Proof}

We prove the proposition by contradiction. Suppose that $J$ is not CA-observable with respect to $L(G)$, $\Sigma _{o}$, $\Sigma ^a_o$ and $\Phi^\pi$. Then
\begin{align*}
	& \neg (\forall w \in \Sigma ^*)(\forall \sigma \in \Sigma) (w \sigma \in J \\
	& \Rightarrow (\exists v \in \Phi ^\pi (w))(\forall w' \in (\Phi ^\pi)^{-1} (v)) \\
	& (w' \in J \wedge w' \sigma \in L(G) \Rightarrow w' \sigma \in J)) \\
	\Leftrightarrow \
	& (\exists w \in \Sigma ^*)(\exists \sigma \in \Sigma) (w \sigma \in J \\
	& \wedge \neg (\exists v \in \Phi ^\pi (w))(\forall w' \in (\Phi ^\pi)^{-1} (v)) \\
	& (w' \in J \wedge w' \sigma \in L(G) \Rightarrow w' \sigma \in J)) \\
	\Leftrightarrow \
	& (\exists w \in \Sigma ^*)(\exists \sigma \in \Sigma) (w \sigma \in J \\
	& \wedge (\forall v \in \Phi ^\pi (w))(\exists w' \in (\Phi ^\pi)^{-1} (v)) \\
	& (w' \in J \wedge w' \sigma \in L(G) \wedge w' \sigma \not\in J)) \\
	\Rightarrow \
	& (\exists w \in \Sigma ^*)(\exists \sigma \in \Sigma) (w \sigma \in J \\
	& \wedge (\exists v \in \Phi ^\pi (w))(\exists w' \in (\Phi ^\pi)^{-1} (v)) \\
	& (w' \in J \wedge w' \sigma \in L(G) \wedge w' \sigma \not\in J)) \\
	& (\mbox{because } (\forall v \in \Phi ^\pi (w) \cdots \Rightarrow (\exists v \in \Phi ^\pi (w) \cdots) \\
	\Leftrightarrow \
	& (\exists w \in \Sigma ^*)(\exists \sigma \in \Sigma) (w \sigma \in J \\
	& \wedge (\exists v \in \Phi ^\pi (w))(\exists w' \in \Sigma ^*) v \in \Phi ^\pi (w') \\
	& (w' \in J \wedge w' \sigma \in L(G) \wedge w' \sigma \not\in J)) \\
	\Leftrightarrow \
	& (\exists w, w' \in \Sigma ^*)(\exists \sigma \in \Sigma) (w \sigma \in J \\
	& (w' \in J \wedge w' \sigma \in L(G) \wedge w' \sigma \not\in J)) \\
	& \wedge \Phi ^\pi(w) \cap \Phi ^\pi(w') \not= \emptyset ) \\
	\Leftrightarrow \
	& (\exists w, w' \in  J)(\exists \sigma \in \Sigma) (w\sigma \in J \wedge w'\sigma \in L(G) \\
	& \wedge w'\sigma \notin J \wedge \Phi ^\pi(w) \cap \Phi ^\pi(w') \not= \emptyset ) \\
	\Leftrightarrow \
	& (\exists w, w' \in  J)(\exists \sigma \in \Sigma) \neg (w\sigma \in J \wedge w'\sigma \in L(G) \\
	& \wedge w'\sigma \notin J \Rightarrow \Phi ^\pi(w) \cap \Phi ^\pi(w')=\emptyset ) \\
	\Leftrightarrow \
	& \neg (\forall w, w' \in  J)(\forall \sigma \in \Sigma) (w\sigma \in J \wedge w'\sigma \in L(G) \\
	& \wedge w'\sigma \notin J \Rightarrow \Phi ^\pi(w) \cap \Phi ^\pi(w')=\emptyset ) . 
\end{align*}
This contradicts the assumption that $J$ is  CA-D-observable with respect to $L(G)$, $\Sigma _{o}$, $\Sigma ^a_o$ and $\Phi^\pi$

Suppose that $J$ is not CA-S-observable with respect to $L(G)$, $\Sigma _{o}$, $\Sigma ^a_o$ and $\Phi^\pi$. Then
\begin{align*}
	& \neg (\forall w \in J)(\forall \sigma \in \Sigma) ((w\sigma \in L(G) \wedge (\forall v \in \Phi ^\pi(w)) \\
	&  (\exists w' \in J) v \in  \Phi ^\pi(w') \wedge w'\sigma \in J) \Rightarrow w\sigma \in J) \\
	\Leftrightarrow \
	& (\exists w \in J)(\exists \sigma \in \Sigma) \neg ((w\sigma \in L(G) \wedge (\forall v \in \Phi ^\pi(w)) \\
	&  (\exists w' \in J) v \in  \Phi ^\pi(w') \wedge w'\sigma \in J) \Rightarrow w\sigma \in J) \\
	\Leftrightarrow \
	& (\exists w \in J)(\exists \sigma \in \Sigma) (w\sigma \in L(G) \wedge (\forall v \in \Phi ^\pi(w)) \\
	&  (\exists w' \in J) v \in  \Phi ^\pi(w') \wedge w'\sigma \in J) \wedge w\sigma \not\in J \\
	\Rightarrow \
	& (\exists w \in J)(\exists \sigma \in \Sigma) (w\sigma \in L(G) \wedge (\exists v \in \Phi ^\pi(w)) \\
	&  (\exists w' \in J) v \in  \Phi ^\pi(w') \wedge w'\sigma \in J) \wedge w\sigma \not\in J \\
	& (\mbox{because } (\forall v \in \Phi ^\pi(w)) \cdots \Rightarrow (\exists v \in \Phi ^\pi(w)) \cdots ) \\
	\Leftrightarrow \
	& (\exists w', w \in  J)(\exists \sigma \in \Sigma) (w'\sigma \in J \wedge w\sigma \in L(G) \\
	& \wedge w\sigma \notin J \wedge \Phi ^\pi(w') \cap \Phi ^\pi(w) \not= \emptyset ) \\
	\Leftrightarrow \
	& (\exists w', w \in  J)(\exists \sigma \in \Sigma) \neg (w'\sigma \in J \wedge w\sigma \in L(G) \\
	& \wedge w\sigma \notin J \Rightarrow \Phi ^\pi(w') \cap \Phi ^\pi(w)=\emptyset ) \\
	\Leftrightarrow \
	& \neg (\forall w', w \in  J)(\forall \sigma \in \Sigma) (w'\sigma \in J \wedge w\sigma \in L(G) \\
	& \wedge w\sigma \notin J \Rightarrow \Phi ^\pi(w') \cap \Phi ^\pi(w)=\emptyset ) .
\end{align*}
This contradicts the assumption that $J$ is  CA-D-observable with respect to $L(G)$, $\Sigma _{o}$, $\Sigma ^a_o$ and $\Phi^\pi$.

\hfill \bull

{Note that the converse of Proposition \ref{proposition1} is not true, that is, $J$ is CA-observable and CA-S-observable does not imply $J$ is CA-D-observable. This is because the two implications $\Rightarrow$ in the proof are strict, that is, $(\exists v \in \Phi ^\pi (w) \cdots \not\Rightarrow (\forall v \in \Phi ^\pi (w) \cdots)$.
Furthermore,
\begin{align*}
	& \neg (\neg CADO \Rightarrow \neg CAO) \wedge \neg (\neg CADO \Rightarrow \neg CASO) \\
	\Leftrightarrow \
	& \neg (CADO \vee \neg CAO) \wedge \neg (CADO \vee \neg CASO) \\
	\Leftrightarrow \
	& \neg CADO \wedge CAO \wedge \neg CADO \wedge CASO \\
	\Leftrightarrow \
	& \neg CADO \wedge CAO \wedge CASO \\
	\Leftrightarrow \
	& \neg (CADO \vee \neg (CAO \wedge CASO)) \\
	\Leftrightarrow \
	& \neg ((CAO \wedge CASO) \Rightarrow CADO) ,
\end{align*}
where CADO (CAO and CASO, respectively) stands for $J$ is CA-D-observable (CA-observable and CA-S-observable, respectively).
}

We assume that, without loss of generality, the specification language $J$ can be generated by a (trim) sub-automaton of $G$ as
\begin{align*}
	H = (X_H, \Sigma, \xi_H, x_0),
\end{align*}
where $X_H \subseteq X$ and $\xi_H = \xi|_{X_H \times \Sigma} \subseteq \xi$. 

{From a supervisor's point of view, after the observation of a string $v \in \Phi ^\pi (L(G))$, the set of all possible states that $G$ may be in is called state estimate, which is defined as}
\begin{align*}
	E^a(v)= & \{ x \in X: (\exists w \in L(G)) \\
	& v \in \Phi ^\pi (w) \wedge \xi (x_o, w)=x \}.
\end{align*}

Two state-estimate-based supervisors are then defined as follows.
\begin{equation} \label{SEBPNS}
	\begin{split}
		\cS_p (v) = & \{\sigma \in \Sigma : (\forall x \in E^a(v)) \\
		& x \in X_H \wedge \xi (x,\sigma) \in X \Rightarrow \xi (x,\sigma) \in X_H \} \\
		\cS_r (v) = & \{\sigma \in \Sigma : (\exists x \in E^a(v)) \\
		& x \in X_H \wedge \xi(x,\sigma) \in X_H \}.
	\end{split}
\end{equation}

Intuitively, $\cS_p$ takes a conservative attitude in the sense that it will enable $\sigma$ if nothing bad can happen. On the other hand, $\cS_r$ takes an optimistic attitude in the sense that it will enable $\sigma$ if something good may happen.   
The following proposition states that $\cS_p$ and $\cS_r$ are equivalent if $J$ is CA-D-observable.

\begin{proposition} \label{Proposition3} 
	
	Consider a discrete event system $G$ under cyber-attacks and a nonempty closed language $J \subseteq L(G)$ generated by $H$. If $J$ is CA-D-observable with respect to $L(G)$, $\Sigma _{o}$, $\Sigma ^a_o$ and $\Phi^\pi$, then for $v \in \Phi ^\pi (L(G))$, $\cS_p (v) = \cS_r (v)$.
	
\end{proposition}
\noindent {\em Proof}

We first prove by contradiction that, if $J$ is CA-D-observable with respect to $L(G)$, $\Sigma _{o}$, $\Sigma ^a_o$ and $\Phi^\pi$, then for all $v \in \Phi ^\pi (L(G))$ and $\sigma \in \Sigma$,
\begin{equation} \label{Prop3}
	\begin{split}
		& (\forall w \in (\Phi ^\pi)^{-1} (v)) w \in J \wedge w\sigma \in L(G) \Rightarrow w\sigma \in J \\
		\Leftrightarrow \
		& (\exists w' \in (\Phi ^\pi)^{-1} (v)) w' \in J \wedge w'\sigma \in J 
	\end{split}
\end{equation}
as follows.

Suppose Equation (\ref{Prop3}) is not true. Clearly, $A \not\Leftrightarrow B$ if and only if either (1) $A  \wedge \neg B$ or (2) $\neg A \wedge  B$. We show contradictions in both cases. In the first case, we have
\begin{align*}
	& (\forall w \in (\Phi ^\pi)^{-1} (v)) w \in J \wedge w\sigma \in L(G) \Rightarrow w\sigma \in J \\
	& \wedge \neg (\exists w' \in (\Phi ^\pi)^{-1} (v)) w' \in J \wedge w'\sigma \in J \\
	\Leftrightarrow \
	& (\forall w \in (\Phi ^\pi)^{-1} (v)) w \in J \wedge w\sigma \in L(G) \Rightarrow w\sigma \in J \\
	& \wedge (\forall w' \in (\Phi ^\pi)^{-1} (v)) w' \in J \Rightarrow w'\sigma \not\in J \\
	\Rightarrow \
	& (\forall w \in (\Phi ^\pi)^{-1} (v)) (w \in J \wedge w\sigma \in L(G) \Rightarrow w\sigma \in J) \\
	& \wedge (w \in J \Rightarrow w\sigma \not\in J) \\
	& (\mbox{let } w'=w) \\
	\Rightarrow \
	& (\forall w \in (\Phi ^\pi)^{-1} (v)) (w \in J \wedge w\sigma \in L(G) \\
	& \Rightarrow w\sigma \in J \wedge w\sigma \not\in J) .
\end{align*}
which is a contradiction. In the second case, we have
\begin{align*}
	& \neg (\forall w \in (\Phi ^\pi)^{-1} (v)) w \in J \wedge w\sigma \in L(G) \Rightarrow w\sigma \in J \\
	& \wedge (\exists w' \in (\Phi ^\pi)^{-1} (v)) w' \in J \wedge w'\sigma \in J \\
	\Leftrightarrow \
	& (\exists w \in (\Phi ^\pi)^{-1} (v)) w \in J \wedge w\sigma \in L(G) \wedge w\sigma \not\in J \\
	& \wedge (\exists w' \in (\Phi ^\pi)^{-1} (v)) w' \in J \wedge w'\sigma \in J \\
	\Rightarrow \
	& (\exists w', w \in  J)(\exists \sigma \in \Sigma) (w'\sigma \in J \wedge w\sigma \in L(G) \\
	& \wedge w\sigma \notin J \wedge \Phi ^\pi(w') \cap \Phi ^\pi(w) \not= \emptyset ) \\
	\Leftrightarrow \
	& \neg (\forall w', w \in  J)(\forall \sigma \in \Sigma) (w'\sigma \in J \wedge w\sigma \in L(G) \\
	& \wedge w\sigma \notin J \Rightarrow \Phi ^\pi(w') \cap \Phi ^\pi(w)=\emptyset ),
\end{align*}
which contradicts the assumption that $J$ is CA-D-observable with respect to $L(G)$, $\Sigma _{o}$, $\Sigma ^a_o$ and $\Phi^\pi$.

We can now prove $\cS_p (v) = \cS_r (v)$ as follows.
\begin{align*}
	& \sigma \in \cS_p (v) \\
	\Leftrightarrow \
	& (\forall x \in E^a(v)) x \in X_H \wedge \xi (x,\sigma) \in X \Rightarrow \xi (x,\sigma) \in X_H \\
	\Leftrightarrow \
	& (\forall w \in (\Phi ^\pi)^{-1} (v)) w \in J \wedge w\sigma \in L(G) \Rightarrow w\sigma \in J \\
	\Leftrightarrow \
	& (\exists w' \in (\Phi ^\pi)^{-1} (v)) w' \in J \wedge w'\sigma \in J \\
	\Leftrightarrow \
	& (\exists x \in E^a(v)) x \in X_H \wedge \xi(x,\sigma) \in X_H \\
	\Leftrightarrow \
	& \sigma \in \cS_r (v) .
\end{align*}

\hfill \bull

The following theorem gives a necessary and sufficient condition for the existence of a deterministic supervisor.

\begin{theorem} \label{Theorem1} 
	
	Consider a discrete event system $G$ under cyber-attacks. For a nonempty closed language $J \subseteq L(G)$ generated by $H$, there exists a deterministic supervisor $\cS$ such that $L_a(\cS^a/G)=L_r(\cS^a/G)=J$ 
	if and only if $J$ is CA-D-controllable with respect to $L(G)$, $\Sigma _{uc}$, and $\Sigma ^a_c$, and CA-D-observable with respect to $L(G)$, $\Sigma _{o}$, $\Sigma ^a_o$ and $\Phi^\pi$. 
	Furthermore, if a deterministic supervisor exists, then $\cS_p$ (= $\cS_r$) defined in  Equation (\ref{SEBPNS}) is such a supervisor.
	
\end{theorem}
\noindent {\em Proof}

By Proposition \ref{proposition1}, $J$ is CA-D-observable implies that it is CA-observable. We first show that if $J$ is CA-observable with respect to $L(G)$, $\Sigma _{o}$, $\Sigma ^a_o$ and $\Phi^\pi$, then, for all $w \in \Sigma ^*$ and $\sigma \in \Sigma$,
\begin{equation} \label{Equation8}
	\begin{split}
		& w \sigma \in J \\
		\Leftrightarrow \
		& w \in J \wedge w \sigma \in L(G) \\
		& \wedge (\exists v \in \Phi ^\pi(w)) (\forall w' \in (\Phi ^\pi)^{-1} (v)) \\
		& (w' \in J \wedge w' \sigma \in L(G) \Rightarrow w' \sigma \in J).
	\end{split}
\end{equation}

The proof is as follows.

($\Rightarrow$) By the definition of CA-observability, Equation
(\ref{CAO}), we have
\begin{align*}
	& w \sigma \in J \\
	\Rightarrow \ & (\exists v \in \Phi ^\pi (w))(\forall w' \in (\Phi ^\pi)^{-1} (v)) \\
	& (w' \in J \wedge w' \sigma \in L(G) \Rightarrow w' \sigma \in J) \\
	\Rightarrow \ & w \in J \wedge w \sigma \in L(G) \wedge (\exists v \in \Phi ^\pi (w)) \\
	& (\forall w' \in (\Phi ^\pi)^{-1} (v)) (w' \in J \wedge w' \sigma \in L(G) \Rightarrow w' \sigma \in J) \\
	& (\mbox{since } w \sigma \in J \Rightarrow w \in J \wedge w \sigma \in L(G))
\end{align*}

($\Leftarrow$) We have
\begin{align*}
	& w \in J \wedge w \sigma \in L(G)\wedge (\exists v \in \Phi ^\pi (w)) \\
	& (\forall w' \in (\Phi ^\pi)^{-1} (v)) (w' \in J \wedge w' \sigma \in L(G) \Rightarrow w' \sigma \in J) \\
	\Rightarrow \ & w \in J \wedge w \sigma \in L(G) \wedge (w \in J \wedge w \sigma \in L(G) \Rightarrow w \sigma \in J) \\
	& (\mbox{take any $v \in \Phi ^\pi (w)$ and let $w'=w$}) \\
	\Rightarrow \ & w \sigma \in J .
\end{align*} 

We now show that if $J$ is CA-D-observable with respect to $L(G)$, $\Sigma _{o}$, $\Sigma ^a_o$ and $\Phi^\pi$, then, for all $w \in \Sigma ^*$ and $\sigma \in \Sigma$,
\begin{equation} \label{Equation4}
	\begin{split}
		& w \sigma \in J \\
		\Leftrightarrow \
		& ( w \in J \wedge w \sigma \in L(G) \wedge (\forall v \in \Phi ^\pi(w)) \\
		& (\forall w' \in (\Phi ^\pi)^{-1} (v)) ( w' \in J \wedge w' \sigma \in L(G) \\
		& \Rightarrow w' \sigma \in J))  .
	\end{split}
\end{equation}

The proof is as follows.

($\Leftarrow$) This implication is clear by taking $w'=w$.

($\Rightarrow$) We prove this implication by contradiction as
follows. If $\Rightarrow$ is not true, then
\begin{align*}
	& (\exists w \in \Sigma ^*)(\exists \sigma \in \Sigma) w \sigma \in J \wedge \neg (\forall v \in \Phi ^\pi(w)) \\
	& (\forall w' \in (\Phi ^\pi)^{-1} (v))  (w' \in J \wedge w' \sigma \in L(G) \Rightarrow w' \sigma \in J) \\
	\Rightarrow \
	& (\exists w \in \Sigma ^*)(\exists \sigma \in \Sigma) w \sigma \in J \wedge (\exists v \in \Phi ^\pi(w)) \\
	& (\exists w' \in (\Phi ^\pi)^{-1} (v)) w' \in J \wedge w' \sigma \in L(G) \wedge w' \sigma \not \in J \\
	\Rightarrow \
	& (\exists w \in \Sigma ^*)(\exists \sigma \in \Sigma) w \sigma \in J \wedge (\exists w' \in \Sigma^*) \\
	& (\exists v \in \Phi ^\pi(w) \cap \Phi ^\pi(w')) w' \in J \wedge w' \sigma \in L(G) \wedge w' \sigma \not \in J \\
	\Rightarrow \
	& (\exists w, w' \in J) (\exists \sigma \in \Sigma) w \sigma \in J \wedge w' \sigma \in L(G) \wedge w' \sigma \not \in J  \\
	& \wedge \Phi ^\pi(w) \cap \Phi ^\pi(w') \not= \emptyset,
\end{align*}
which contradicts the CA-D-observability of $J$.

Note that
\begin{equation} \label{gamma}
	\begin{split}
		& (\exists \gamma \in \cS ^a(v)) \sigma \in \gamma \\
		\Leftrightarrow \
		& (\exists \gamma \in \Delta (\cS (v))) \sigma \in \gamma \\
		\Leftrightarrow \
		& (\exists \gamma ', \gamma '' \in \Sigma _c^a) \sigma \in (\cS (v)-\gamma ') \cup \gamma '' \\
		\Leftrightarrow \
		& \sigma \in \cS (v) \cup \Sigma _c^a \\
		& (\mbox{since } \gamma '= \emptyset \wedge \gamma '' = \Sigma _c^a \mbox{ covers all cases}) .
	\end{split}
\end{equation}

Therefore, 
\begin{equation} \label{Equation10}
	\begin{split}
		&  w \sigma\in L_a (\cS ^a/G) \\
		\Leftrightarrow \
		& w \in L_a (\cS ^a/G) \wedge w \sigma \in L(G) \wedge (\sigma \in \Sigma _{uc}  \\
		& \vee (\exists v \in \Phi ^\pi (w))(\exists \gamma \in \cS ^a(v)) \sigma \in \gamma) \\
		& (\mbox{by Equation (\ref{LargeL})}) \\
		\Leftrightarrow \
		& w \in L_a (\cS ^a/G) \wedge w \sigma \in L(G) \wedge (\sigma \in \Sigma _{uc} \\
		& \vee (\exists v \in \Phi ^\pi (w)) \sigma \in \cS (v) \cup \Sigma _c^a ) \\
		& (\mbox{by Equation (\ref{gamma})}) \\
		\Leftrightarrow \
		& w \in L_a (\cS ^a/G) \wedge w \sigma \in L(G) \wedge (\sigma \in \Sigma _{uc} \cup \Sigma _c^a \\
		& \vee (\exists v \in \Phi ^\pi (w)) \sigma \in \cS (v)) .
	\end{split}
\end{equation}

Note also that
\begin{equation} \label{allgamma}
	\begin{split}
		& (\forall \gamma \in  \cS^a(v))\sigma \in \gamma  \\
		\Leftrightarrow \
		& (\forall \gamma \in \Delta (\cS(v)))\sigma \in \gamma \\
		\Leftrightarrow \
		& (\forall \gamma', \gamma'' \subseteq \Sigma^a_c) \sigma \in((\cS(v)-\gamma')\cup \gamma'') \\
		\Leftrightarrow \
		& \sigma \in \cS(v) - \Sigma^a_c \\
		& (\mbox{since } \gamma' = \Sigma^a_c \wedge \gamma'' = \emptyset \mbox{ covers all cases}). 
	\end{split}
\end{equation}

Furthermore,
\begin{equation} \label{MBcc}
	\begin{split}
		& w \sigma \in L_r(\cS^a/G) \\
		\Leftrightarrow \
		& w\in L_r(\cS^a/G)\wedge w \sigma \in L(G) \wedge \sigma \notin \Sigma^a_c \\
		& \wedge (\sigma \in \Sigma _{uc} \vee (\forall v \in  \Phi ^\pi(w)) \sigma \in \cS(v) ) ,
	\end{split}
\end{equation}
which is derived as follows.
\begin{align*}
	& w \sigma \in L_r(\cS^a/G) \\
	\Leftrightarrow \
	& w\in L_r(\cS^a/G)\wedge w \sigma \in L(G) \wedge (\sigma \in \Sigma_{uc}  \\
	& \vee  (\forall v \in   \Phi ^\pi(w))(\forall \gamma \in  \cS^a(v))\sigma \in \gamma) \\
	& (\mbox{by Equation (\ref{SmallL})}) \\
	\Leftrightarrow \
	& w\in L_r(\cS^a/G)\wedge w \sigma \in L(G) \wedge (\sigma \in \Sigma _{uc} \\
	& \vee (\forall v \in \Phi ^\pi(w)) \sigma \in \cS(v) - \Sigma^a_c)  \\
	& (\mbox{by Equation (\ref{allgamma})}) \\
	\Leftrightarrow \
	& w\in L_r(\cS^a/G)\wedge w \sigma \in L(G) \wedge (\sigma \in \Sigma _{uc}\\
	& \vee (\forall v \in \Phi ^\pi(w)) (\sigma \in \cS(v) \wedge \sigma \notin \Sigma^a_c)) \\
	\Leftrightarrow \
	& w\in L_r(\cS^a/G)\wedge w \sigma \in L(G) \wedge (\sigma \in \Sigma _{uc} \\
	& \vee (\sigma \notin \Sigma^a_c \wedge (\forall v \in \Phi ^\pi(w)) \sigma \in \cS(v) )) \\
	\Leftrightarrow \
	& w\in L_r(\cS^a/G)\wedge w \sigma \in L(G) \\
	& \wedge ((\sigma \in \Sigma _{uc} \vee \sigma \notin \Sigma^a_c) \\
	& \wedge (\sigma \in \Sigma _{uc} \vee (\forall v \in \Phi ^\pi(w)) \sigma \in \cS(v) )) \\
	\Leftrightarrow \
	& w\in L_r(\cS^a/G)\wedge w \sigma \in L(G) \wedge \sigma \notin \Sigma^a_c \\
	& \wedge (\sigma \in \Sigma _{uc} \vee (\forall v \in  \Phi ^\pi(w)) \sigma \in \cS(v) ) \\
	& (\mbox{because } \sigma \in \Sigma _{uc} \Rightarrow \sigma \notin \Sigma^a_c ).	
\end{align*}

From Equation (\ref{MBcc}), it is clear that
\begin{align*}
	& w \sigma \in L_r(\cS^a/G) \\
	\Leftrightarrow \
	& w\in L_r(\cS^a/G)\wedge w \sigma \in L(G) \wedge \sigma \notin \Sigma^a_c \\
	& \wedge (\sigma \in \Sigma _{uc} \vee (\forall v \in  \Phi ^\pi(w)) \sigma \in \cS(v) ) \\
	\Rightarrow \
	& (\sigma \in \Sigma _{uc} \vee (\forall v \in  \Phi ^\pi(w)) \sigma \in \cS(v) ) .
\end{align*}

Therefore,
\begin{equation} \label{MBccneg}
	\begin{split}
		& \neg (\sigma \in \Sigma _{uc} \vee (\forall v \in  \Phi ^\pi(w)) \sigma \in \cS(v) ) \\
		\Rightarrow \
		& w \sigma \not\in L_r(\cS^a/G) .
	\end{split}
\end{equation}


(IF) We now prove the ``if'' part. The assumption is that language $J$ is both CA-D-controllable with respect to $L(G)$, $\Sigma _{uc}$, and $\Sigma ^a_c$, and CA-D-observable with respect to $L(G)$, $\Sigma _{o}$, $\Sigma ^a_o$, and $\Phi^\pi$.

According to Proposition \ref{proposition2}, this implies that $J$ is: (1) CA-controllable with respect to $L(G)$, $\Sigma _{uc}$, and $\Sigma ^a_c$, (2) CA-S-controllable with respect to $L(G)$, $\Sigma _{uc}$, and $\Sigma ^a_c$, and (3) CA-D-observable with respect to $L(G)$, $\Sigma _{o}$, $\Sigma ^a_o$, and $\Phi^\pi$. 
We will demonstrate that the supervisor $\cS_p$, defined in Equation (\ref{SEBPNS}), is a deterministic supervisor satisfying $L_a(\cS_p^a/G) = L_r(\cS_p^a/G) = J$. Specifically,
$$
J \subseteq L_r(\cS_p^a/G) \subseteq L_a(\cS_p^a/G)
\subseteq J.
$$

Since $L_r(\cS_p^a/G) \subseteq L_a(\cS_p^a/G)$ by the definitions of $L_r(\cS_p^a/G)$ and $L_a(\cS_p^a/G)$, we only need to prove 
\begin{align*}
	(A) \ \ & J \subseteq 	L_r(\cS_p^a/G) \\
	(B) \ \ & L_a(\cS_p^a/G) \subseteq J.
\end{align*}

We prove (A) $w \in J \Rightarrow w \in L_r(\cS_p^a/G)$ by
induction on the length $|w|$ of $w \in \Sigma ^*$ as follows.

{\em Base:} Since $J$ is nonempty and closed, $\varepsilon \in J$. By
definition, $\varepsilon \in L_r(\cS_p^a/G)$. Therefore, for
$|w| = 0$, that is, $w=\varepsilon$, we have
$$
w \in J \Rightarrow w \in L_r(\cS_p^a/G).
$$

{\em Induction Hypothesis:} Assume that for all $w \in \Sigma^*$,
$|w| \leq n$,
$$
w \in J \Rightarrow w \in L_r(\cS_p^a/G).
$$

{\em Induction Step:} We show that for all $w \in \Sigma^*$,
$\sigma \in \Sigma$, $|w\sigma| = n+1$,
$$
w \sigma \in J \Rightarrow w \sigma \in L_r(\cS_p^a/G).
$$
as follows.
\begin{align*}
	& w \sigma \in J \\
	\Rightarrow \
	& w \sigma \in J \wedge \sigma \notin \Sigma^a_c \\
	& (\mbox{because $J$ is CA-S-controllable }) \\
	\Rightarrow \
	& w \in J \wedge w \sigma \in L(G) \wedge w \sigma \in J \wedge \sigma \notin \Sigma^a_c \\
	\Rightarrow  \
	& w \in L_r(\cS_p^a/G) \wedge w \sigma \in L(G) \wedge w \sigma \in J \wedge \sigma \notin \Sigma^a_c \\
	& (\mbox{by Induction Hypothesis}) \\
	\Rightarrow \
	& w \in L_r(\cS_p^a/G) \wedge w \sigma \in L(G) \wedge \sigma \notin \Sigma^a_c \\
	& \wedge (\forall v \in \Phi ^\pi(w))(\forall w' \in (\Phi ^\pi)^{-1} (v)) \\
	& (w' \in J \wedge w' \sigma \in L(G)\Rightarrow w' \sigma \in J) \\
	& (\mbox{by Equation (\ref{Equation4})}) \\
	\Rightarrow \ 
	& w \in L_r(\cS_p^a/G) \wedge w \sigma \in L(G) \wedge \sigma \notin \Sigma^a_c \\
	& \wedge (\forall v \in \Phi ^\pi(w)) (\forall x \in E^a(v)) \\
	& (x \in X_H \wedge \xi (x,\sigma) \in X \Rightarrow \xi (x,\sigma) \in X_H) \\
	& (\mbox{by the definition of } E^a ) \\
	\Rightarrow \ 
	& w \in L_r(\cS_p^a/G) \wedge w \sigma \in L(G) \wedge \sigma \notin \Sigma^a_c \\
	& \wedge (\forall v \in \Phi ^\pi(w)) \sigma \in \cS_p (v) \\
	& (\mbox{by the definition of } \cS_p  ) \\
	\Rightarrow \
	& w\in L_r(\cS_p^a/G)\wedge w \sigma \in L(G) \wedge \sigma \notin \Sigma^a_c \\
	& \wedge (\sigma \in \Sigma _{uc} \vee (\forall v \in   \Phi ^\pi(w)) \sigma \in \cS_p(v) ) \\
	\Rightarrow  & w \sigma \in L_r(\cS_p^a/G) \\
	& (\mbox{by Equation (\ref{MBcc})}) .
\end{align*}

We prove (B) $w \in L_a(\cS_p^a/G) \Rightarrow w \in J$ by
induction on the length $|w|$ of $w \in \Sigma ^*$ as follows.

{\em Base:} Since $J$ is nonempty and closed, $\varepsilon \in J$. By
definition, $\varepsilon \in L_a(\cS_p^a/G)$. Therefore, for
$|w| = 0$, that is, $w=\varepsilon$, we have

$$
w \in L_a(\cS_p^a/G) \Rightarrow w \in J.
$$

{\em Induction Hypothesis:} Assume that for all $w \in \Sigma^*$,
$|w| \leq n$,
$$
w \in L_a(\cS_p^a/G) \Rightarrow w \in J.
$$

{\em Induction Step:} We show that for all $w \in \Sigma^*$,
$\sigma \in \Sigma$, $|w\sigma| = n+1$,
$$
w \sigma \in L_a(\cS_p^a/G) \Rightarrow w \sigma \in J
$$
as follows.
\begin{align*}
	& w \sigma \in L_a(\cS_p^a/G) \\
	\Rightarrow \ 
	& w \in L_a (\cS_p^a/G) \wedge w \sigma \in L(G) \wedge (\sigma \in \Sigma _{uc} \cup \Sigma _c^a \\
	& \vee (\exists v \in \Phi ^\pi (w)) \sigma \in \cS_p (v)) \\
	& (\mbox{by Equation (\ref{Equation10})}) \\
	\Rightarrow \ 
	& w \in J \wedge w \sigma \in L(G) \wedge (\sigma \in \Sigma _{uc} \cup \Sigma _c^a \\
	& \vee (\exists v \in \Phi ^\pi (w)) \sigma \in \cS_p (v)) \\
	& (\mbox{by Induction Hypothesis}) \\
	\Rightarrow \
	& (w \in J \wedge w \sigma \in L(G) \wedge \sigma \in \Sigma _{uc} \cup \Sigma _c^a) \vee ( w \in J \\
	& \wedge w \sigma \in L(G) \wedge (\exists v \in \Phi ^\pi(w))
	\sigma \in \cS_p^a (v)) \\
	\Rightarrow \ & w \sigma \in J \vee ( w \in J \wedge w \sigma
	\in L(G) \\
	& \wedge (\exists v \in \Phi ^\pi(w)) \sigma \in \cS_p
	^a (v)) \\
	& (\mbox{because $J$ is CA-controllable }) \\
	\Rightarrow \
	& w \sigma \in J \vee ( w \in J \wedge w \sigma
	\in L(G) \\
	& \wedge (\exists v \in \Phi ^\pi(w))(\forall x \in
	E^a(v)) \\
	& (x \in X_H \wedge \xi (x,\sigma) \in X \Rightarrow \xi
	(x,\sigma) \in X_H) \\
	& (\mbox{by the definition of } \cS_p^a ) \\
	\Rightarrow \
	& w \sigma \in J \vee ( w \in J \wedge w \sigma
	\in L(G) \\
	& \wedge (\exists v \in \Phi ^\pi(w))(\forall w' \in
	(\Phi ^\pi)^{-1} (v)) \\
	& (w' \in J \wedge w' \sigma \in L(G) \Rightarrow w' \sigma \in J) \\
	& (\mbox{by the definition of } E^a ) \\
	\Rightarrow \
	& w \sigma \in J \vee w \sigma \in J \\
	& (\mbox{by Equation (\ref{Equation8})}) \\
	\Rightarrow \
	& w \sigma \in J .
\end{align*}

(ONLY IF) Let us prove the ``only if'' part. Assume that there exists a deterministic supervisor $\cS$ such that $L_a(\cS^a/G)=L_r(\cS^a/G)=J$. Since $L_a(\cS^a/G)=J$, by Theorem 3 of \cite{zheng2024modeling}, $J$ is CA-controllable with respect to $L(G)$, $\Sigma _{uc}$, and $\Sigma ^a_c$. 
Since $L_r(\cS^a/G)=J$, by Theorem 2 of \cite{wang2025supervisory}, $J$ is CA-S-controllable with respect to $L(G)$, $\Sigma _{uc}$, and $\Sigma ^a_c$. 
Therefore, by Proposition \ref{proposition2}, $J$ is CA-D-controllable with respect to $L(G)$, $\Sigma _{uc}$, and $\Sigma ^a_c$. 

Next, we prove that $J$ is CA-D-observable with respect to $L(G)$, $\Sigma _{o}$, $\Sigma ^a_o$ and $\Phi^\pi$ by contradiction. Suppose that $J$ is not CA-D-observable, then
\begin{align*}
	& \neg (\forall w, w' \in  J)(\forall \sigma \in \Sigma) (w\sigma \in J \wedge w'\sigma \in L(G) \\
	& \wedge w'\sigma \notin J \Rightarrow \Phi ^\pi(w) \cap \Phi ^\pi(w')=\emptyset ) \\
	\Leftrightarrow \
	& (\exists w, w' \in  J)(\exists \sigma \in \Sigma) (w\sigma \in J \wedge w'\sigma \in L(G) \\
	& \wedge w'\sigma \notin J \wedge \Phi ^\pi(w) \cap \Phi ^\pi(w') \not= \emptyset ) \\
	\Leftrightarrow \
	& (\exists w, w' \in  J)(\exists \sigma \in \Sigma) (w\sigma \in J \wedge w'\sigma \in L(G) \\
	& \wedge w'\sigma \notin J \wedge (\exists v \in \Sigma^*) v \in \Phi ^\pi(w) \wedge v \in \Phi ^\pi(w') ) .
\end{align*}
Since $J$ is CA-D-controllable,
\begin{align*}
	& w' \in  J \wedge w'\sigma \in L(G) \wedge w'\sigma \notin J \Rightarrow \sigma \not\in \Sigma _{uc} \cup \Sigma ^a_c .
\end{align*}
Consider the following two possible cases. 

\emph{Case 1.} $\sigma \in \cS(v)$: In this case, 
\begin{align*}
	& w'\sigma \notin J \wedge w' \in  J \wedge w'\sigma \in L(G) \\
	& \wedge (\exists v \in \Phi ^\pi(w')) \sigma \in \cS(v) \\
	\Leftrightarrow \
	& w'\sigma \notin J \wedge w' \in L_a (\cS ^a/G) \wedge w'\sigma \in L(G) \\
	& \wedge (\exists v \in \Phi ^\pi(w')) \sigma \in \cS(v) \\
	& (\mbox{because } L_a(\cS^a/G)=J ) \\
	\Rightarrow \
	& w'\sigma \notin J \wedge w' \in L_a (\cS ^a/G) \wedge w' \sigma \in L(G) \\
	& \wedge (\sigma \in \Sigma _{uc} \cup \Sigma _c^a \vee (\exists v \in \Phi ^\pi (w')) \sigma \in \cS (v)) \\
	\Leftrightarrow \
	&  w'\sigma \notin J \wedge w' \sigma\in L_a (\cS ^a/G) \\
	& (\mbox{by Equation (\ref{Equation10})}) ,
\end{align*}
which contradicts the assumption that $L_a(\cS^a/G)=J$.

\emph{Case 2.} $\sigma \not\in \cS(v)$: In this case, 
\begin{align*}
	& w\sigma \in J \wedge \sigma \not\in \Sigma _{uc} \cup \Sigma ^a_c \wedge (\exists v \in \Phi ^\pi(w)) \sigma \not\in \cS(v) \\
	\Rightarrow \
	& w\sigma \in J \wedge \sigma \not\in \Sigma _{uc} \wedge (\exists v \in \Phi ^\pi(w)) \sigma \not\in \cS(v) \\
	\Rightarrow \
	& w\sigma \in J \wedge \neg (\sigma \in \Sigma _{uc} \vee (\forall v \in  \Phi ^\pi(w)) \sigma \in \cS(v) ) \\
	\Rightarrow \
	& w\sigma \in J \wedge w \sigma \not\in L_r(\cS^a/G) \\
	& (\mbox{by Equation (\ref{MBccneg})}),
\end{align*}
which contradicts $J=L_r(\cS^a/G)$.

\hfill \bull

{Note that the existence condition for the deterministic control cannot be obtained by taking the conjunction of the existence conditions for the large-language control and the small-language control. Even if the conjunction is satisfied, the supervisor $S_1$ achieving $L_a (S_1^a/G)=J$ and the supervisor $S_2$ achieving $L_r (S_2^a/G)=J$ may not be the same. Therefore, a stronger condition is needed to ensure the existence of the deterministic control. This is reflected in the fact that CA-D-observability is strictly stronger than the conjunction of CA-observability and CA-S-observability.}

{Note further that, since CA-D-controllability and CA-S-controllability require that $J \subseteq (\Sigma - \Sigma^a_c)^*$, if $J$ contains all events in $\Sigma$, then it is necessary that $\Sigma^a_c = \emptyset$. In other words, for deterministic control (and small-language control) to exist, it is required that attackable controllable events cannot appear in $J$. This is restrictive, but necessary, because CA-D-controllability (and CA-S-controllability), which is a part of the necessary (and sufficient) condition, requires it. 
}

\section{Nonblocking Supervisory Control}
\label{Section4}

In conventional supervisory control, the {\em language marked} by closed-loop system $\cS/G$ is defined as
$$
L_m(\cS/G) = L(\cS/G) \cap L_m(G). 
$$

We say that a supervisor $S$ is \emph{nonblocking} if 
$$
\overline{L_m(\cS/G)} = L(\cS/G). 
$$

Given a non-closed language $K \subseteq L_m(G)$, we say $K$ is $L_m(G)$-closed if
$$
K = \overline{K} \cap L_m(G). 
$$
If $K$ is not closed, we say $K$ is controllable, observable, CA-, CA-S-, CA-D-controllable, and CA-, CA-S, CA-D-observable, if $\overline{K}$ is controllable, observable, CA-, CA-S-, CA-D-controllable, and CA-, CA-S, CA-D-observable, respectively.

For conventional supervisory control, it is proved in \cite{lin1988observability} that there exists a nonblocking supervisor such that $L_m(\cS/G)=K$, if and only if $K$ is $L_m(G)$-closed, controllable with respect to $L(G)$ and $\Sigma_{uc}$, and observable with respect to $L(G)$ and $\Sigma_o$.

For supervisory control under cyber-attacks, the definition of nonblocking is more complex, because $L(\cS/G)$ is nondeterministic. Furthermore, $L_m(G)$-closeness is not enough to ensure nonblocking. Some additional conditions are needed. Define the marked large language as
$$
L_{am}(\cS^a/G) = L_a(\cS^a/G) \cap L_m(G).
$$

A supervisor $\cS :\Phi ^\pi (L(G)) \rightarrow \Gamma$ under cyber-attacks is nonblocking if
\begin{enumerate}
	\item
	$\overline{L_{am}(\cS^a/G)} =L_a(\cS^a/G)$
	\item
	$(\forall w \in L_a(\cS^a/G)) (\forall \sigma \in \Sigma_c)w \sigma \in L_a(\cS^a/G)$ \\
	$\Rightarrow (\forall v \in \Phi ^\pi (w))(\forall \gamma \in \cS ^a(v)) \sigma \in \gamma$.
\end{enumerate}
The second condition above says that if $w \sigma$ is allowed in $L_a(\cS^a/G)$, then $\sigma$ must be enabled by all possible controls after $w$. This is needed because, due to nondeterminism in control, although $w \sigma \in L_a(\cS^a/G)$, $\sigma$ may still be blocked by some controls.

Let us prove the following theorem.

\begin{theorem} (Second Condition for Nonblocking)
	\label{theorem2}
	
	Consider a discrete event system $G$ under cyber-attacks. For any supervisor $\cS :\Phi ^\pi (L(G)) \rightarrow \Gamma$,
	\begin{equation} \label{E8}
		\begin{split}
			& (\forall w \in L_a(\cS^a/G)) (\forall \sigma \in \Sigma_c)w \sigma \in L_a(\cS^a/G) \\
			& \Rightarrow (\forall v \in \Phi ^\pi (w))(\forall \gamma \in \cS ^a(v)) \sigma \in \gamma
		\end{split}
	\end{equation}
	if and only if
	\begin{equation} \label{E9}
		\begin{split}
			& L_a(\cS^a/G) = L_r(\cS^a/G) .
		\end{split}
	\end{equation}
	
\end{theorem}
\noindent {\em Proof} 

(IF) Let us prove the ``if'' part by contradiction. Suppose Equation (\ref{E8}) is not true. Then
\begin{align*}
	& \neg (\forall w \in L_a(\cS^a/G)) (\forall \sigma \in \Sigma_c)w \sigma \in L_a(\cS^a/G) \\
	& \Rightarrow (\forall v \in \Phi ^\pi (w))(\forall \gamma \in \cS ^a(v)) \sigma \in \gamma \\
	\Leftrightarrow \
	& (\exists w \in L_a(\cS^a/G)) (\exists \sigma \in \Sigma_c)w \sigma \in L_a(\cS^a/G) \\
	& \wedge \neg (\forall v \in \Phi ^\pi (w))(\forall \gamma \in \cS ^a(v)) \sigma \in \gamma \\
	\Rightarrow \
	& (\exists w \in L_a(\cS^a/G)) (\exists \sigma \in \Sigma_c)w \sigma \in L_a(\cS^a/G) \\
	& \wedge w \sigma \not\in L_r (\cS^a/G) \\
	& (\mbox{by Equation (\ref{SmallL})}) ,
\end{align*}
which contradicts $L_a(\cS^a/G) = L_r(\cS^a/G)$. 

(ONLY IF) Let us prove the ``only if'' part by contradiction. Suppose Equation (\ref{E9}) is not true. Then
\begin{align*}
	& L_a(\cS^a/G) \not= L_r(\cS^a/G) \\
	\Leftrightarrow \
	& L_a(\cS^a/G) \not\subseteq L_r(\cS^a/G) \vee L_r(\cS^a/G) \not\subseteq L_a(\cS^a/G) \\
	\Leftrightarrow \
	& L_a(\cS^a/G) \not\subseteq L_r(\cS^a/G)  \\
	& (\mbox{because } L_r(\cS^a/G) \subseteq L_a(\cS^a/G)) \\
	\Leftrightarrow \
	& (\exists w \in L_a(\cS^a/G) \cap L_r(\cS^a/G)) (\exists \sigma \in \Sigma_c) \\
	& w \sigma \in L_a(\cS^a/G) \wedge w \sigma \not\in  L_r(\cS^a/G)  \\
	\Rightarrow \
	& (\exists w \in L_a(\cS^a/G)) (\exists \sigma \in \Sigma_c)w \sigma \in L_a(\cS^a/G) \\
	& \wedge \neg (\forall v \in \Phi ^\pi (w))(\forall \gamma \in \cS ^a(v)) \sigma \in \gamma \\
	& (\mbox{Equation (\ref{SmallL})}) \\
	\Leftrightarrow \
	& \neg (\forall w \in L_a(\cS^a/G)) (\forall \sigma \in \Sigma_c)w \sigma \in L_a(\cS^a/G) \\
	& \Rightarrow (\forall v \in \Phi ^\pi (w))(\forall \gamma \in \cS ^a(v)) \sigma \in \gamma ,
\end{align*}
which contradicts Equation (\ref{E8}).

\hfill \bull

We assume that, without loss of generality, the non-closed language $K \subseteq L_m(G)$ is marked by a trim sub-automaton of $G$ as
\begin{align*}
	H = (X_H, \Sigma, \xi_H, x_0, X_{m,H}).
\end{align*}
In other words, let us add marked states $X_{m,H}$ to $H$ such that $K=L_m(H)$. We now prove the following theorem.

\begin{theorem} (Existence Condition for Nonblocking Supervisor)
	\label{theorem3}
	
	Consider a discrete event system $G$ under cyber-attacks. For a nonempty language $K \subseteq L_m(G)$ marked by $H$, there exists a nonblocking supervisor $\cS$ such that ${L_{am}(\cS^a/G)} = {K}$ 
	if and only if $K$ is $L_m(G)$-closed, CA-D-controllable with respect to $L(G)$, $\Sigma _{uc}$, and $\Sigma ^a_c$, and CA-D-observable with respect to $L(G)$, $\Sigma _{o}$, $\Sigma ^a_o$ and $\Phi^\pi$. 
	Furthermore, if a nonblocking supervisor exists, then $\cS_p$ (= $\cS_r$) defined in  Equation (\ref{SEBPNS}) is such a supervisor.
	
\end{theorem}
\noindent {\em Proof} 

(IF) Assume that $K$ is $L_m(G)$-closed, and $K$  (and hence $\overline{K}$) is CA-D-controllable with respect to $L(G)$, $\Sigma _{uc}$, and $\Sigma ^a_c$, and CA-D-observable with respect to $L(G)$, $\Sigma _{o}$, $\Sigma ^a_o$ and $\Phi^\pi$. Then, by Theorem \ref{Theorem1}, $\cS_p$ defined in Equation (\ref{SEBPNS}) is a deterministic supervisor such that $L_a(\cS^a/G) =L_r(\cS^a/G) =\overline{K}$. Therefore,
$$
{L_{am}(\cS^a/G)}  = {L_a(\cS^a/G) \cap L_m(G)} = {\overline{K} \cap L_m(G)} = K.
$$
By Theorem \ref{theorem2}, $L_a(\cS^a/G) =L_r(\cS^a/G)$ implies that the second condition for nonblocking is satisfied. The first condition for nonblocking is also satisfied, because
\begin{align*}
	& {L_{am}(\cS^a/G)} = K \\
	\Rightarrow \ & \overline{L_{am}(\cS^a/G)} = \overline{K} = L_a(\cS^a/G).
\end{align*}

(ONLY IF) Assume that there exists a nonblocking supervisor $\cS$ such that ${L_{am}(\cS^a/G)} = {K}$. By Theorem \ref{theorem2}, the second condition for nonblocking implies $L_a(\cS^a/G) =L_r(\cS^a/G)$. 
The first condition for nonblocking implies $L_a(\cS^a/G) =\overline{L_{am}(\cS^a/G)} =\overline{K}$. 
By Theorem \ref{Theorem1}, $\overline{K}$ (and hence $K$) is CA-D-controllable with respect to $L(G)$, $\Sigma _{uc}$, and $\Sigma ^a_c$, and CA-D-observable with respect to $L(G)$, $\Sigma _{o}$, $\Sigma ^a_o$ and $\Phi^\pi$. $K$ is also $L_m(G)$-closed, because
$$
K= {L_{am}(\cS^a/G)}  = {L_a(\cS^a/G) \cap L_m(G)} = {\overline{K} \cap L_m(G)} .
$$

\hfill \bull

{Theorems \ref{theorem2} and \ref{theorem3} say that deterministic control and $L_m(G)$-closeness are necessary and sufficient for nonblocking control. Since nonblocking control is important in many applications of discrete event systems, deterministic control is also important. }

\begin{theorem} (Checking $L_m(G)$-closeness)
	\label{Theorem6}
	
	For a nonempty language $K \subseteq L_m(G)$ marked by $H$, $K$ is $L_m(G)$-closed if and only if
	$$
	X_{m,H} = X_H \cap X_m.
	$$
	
\end{theorem}
\noindent {\em Proof} 
\begin{align*}
	& \mbox{$K$ is $L_m(G)$-closed} \\
	\Leftrightarrow \ & K = {\overline{K} \cap L_m(G)} \\
	\Leftrightarrow \ & (\forall w \in \Sigma^*) w \in K \Leftrightarrow w \in {\overline{K} \cap L_m(G)} \\
	\Leftrightarrow \ & (\forall x \in X) x \in X_{m,H} \Leftrightarrow x \in X_H \cap X_m \\
	& (\mbox{because $G$ and $H$ are trim}) \\	
	\Leftrightarrow \ & X_{m,H} = X_H \cap X_m.
\end{align*}
\hfill \bull

\section{Checking CA-D-controllability and CA-D-observability} \label{s5}

Both deterministic supervisory control and nonblocking supervisory control require CA-D-controllability and CA-D-observability. Let us investigate how to check CA-D-controllability and CA-D-observability in this section.

Checking whether $J$ is {CA-D-controllable} with respect to $L(G)$, $\Sigma _{uc}$, and $\Sigma ^a_c$ is straightforward, because of the following. 
(1) Checking $J (\Sigma _{uc} \cup \Sigma ^a_c) \cap L(G) \subseteq J$ can be done using existing methods for checking controllability by letting the uncontrollable events be $\Sigma _{uc} \cup \Sigma ^a_c$. 
(2) Checking $J \subseteq (\Sigma - \Sigma^a_c)^*$ is also easy.

Checking whether $J$ is {CA-D-observable} with respect to $L(G)$, $\Sigma _{o}$, $\Sigma ^a_o$ and $\Phi^\pi$ is more complex. Let us first prove the following theorem.

\begin{theorem} (Checking CA-D-observability)
	\label{Theorem5} 
	
	$J$ is CA-D-observable with respect to $L(G)$, $\Sigma _{o}$, $\Sigma ^a_o$ and $\Phi^\pi$ if and only if for all $v \in \Phi ^\pi (L(G))$ and $\sigma \in \Sigma$, Equation (\ref{Prop3}) is satisfied, that is, 
	\begin{align*}
		& (\forall w \in (\Phi ^\pi)^{-1} (v)) w \in J \wedge w\sigma \in L(G) \Rightarrow w\sigma \in J \\
		\Leftrightarrow \
		& (\exists w' \in (\Phi ^\pi)^{-1} (v)) w' \in J \wedge w'\sigma \in J .
	\end{align*}
	
\end{theorem}
\noindent {\em Proof} 

(ONLY IF) This part is Proposition \ref{Proposition3}.

(IF) Let us prove the ``if'' part by contradiction. Suppose that $J$ is not CA-D-observable with respect to $L(G)$, $\Sigma _{o}$, $\Sigma ^a_o$ and $\Phi^\pi$, then
\begin{align*}
	& \neg (\forall w', w \in  J)(\forall \sigma \in \Sigma) (w'\sigma \in J \wedge w\sigma \in L(G) \\
	& \wedge w\sigma \notin J \Rightarrow \Phi ^\pi(w') \cap \Phi ^\pi(w)=\emptyset ) \\
	\Leftrightarrow \
	& (\exists w', w \in  J)(\exists \sigma \in \Sigma) (w'\sigma \in J \wedge w\sigma \in L(G) \\
	& \wedge w\sigma \notin J \wedge \Phi ^\pi(w') \cap \Phi ^\pi(w) \not= \emptyset ) \\
	\Rightarrow \
	& (\exists w', w \in  J)(\exists \sigma \in \Sigma) (w'\sigma \in J \wedge w\sigma \in L(G) \wedge w\sigma \notin J \\
	& \wedge (\exists v \in \Phi ^\pi (L(G)) ) v \in \Phi ^\pi(w') \wedge v \in \Phi ^\pi(w) .
\end{align*}

Therefore, there exists $v \in \Phi ^\pi (L(G))$ and $\sigma \in \Sigma$,
\begin{align*}
	& (\exists w', w \in  J) (w'\sigma \in J \wedge w\sigma \in L(G) \wedge w\sigma \notin J \\
	& \wedge v \in \Phi ^\pi(w') \wedge v \in \Phi ^\pi(w) \\
	\Rightarrow \
	& (\exists w \in (\Phi ^\pi)^{-1} (v)) w \in J \wedge w\sigma \in L(G) \wedge w\sigma \not\in J \\
	& \wedge (\exists w' \in (\Phi ^\pi)^{-1} (v)) w' \in J \wedge w'\sigma \in J \\
	\Leftrightarrow \
	& \neg (\forall w \in (\Phi ^\pi)^{-1} (v)) w \in J \wedge w\sigma \in L(G) \Rightarrow w\sigma \in J \\
	& \wedge (\exists w' \in (\Phi ^\pi)^{-1} (v)) w' \in J \wedge w'\sigma \in J ,
\end{align*}
which contradicts Equation (\ref{Prop3}).

\hfill \bull

Based on the above theorem, we propose the following procedure to check CA-D-observability. 

\emph{Step 1}. {For each attackable transition $tr\in \xi^a$, let us assume, without loss of generality, that $A_{tr}$ is marked by a deterministic automaton} 
$$
F_{tr} = (X_{tr} , \Sigma, \xi_{tr}, x_{0,tr}, X_{m,tr}).
$$

We replace a transition $tr=(x, \sigma, x') \in \xi$ in $G$ by $F_{tr}$ as follows.
\begin{align*}
	G_{tr \rightarrow F_{tr}} = ( X \cup X_{tr}, \Sigma , \xi_{tr \rightarrow F_{tr}}, x_0 )
\end{align*}
where $\xi_{tr \rightarrow F_{tr}} = (\xi - \{ (x, \sigma, x') \} ) \cup \xi_{tr} \cup \{ (x, \varepsilon, x_{0,tr}) \} \cup \{ (x_{m,tr}, \varepsilon, x'): x_{m,tr} \in X_{m,tr} \}$. 

Denote the automaton after replacing all attackable transitions in $G$ as
$$
G^\pi=( X^\pi, \Sigma , \xi^ \pi , x_0, X^\pi_m)=( X \cup \hat{X}, \Sigma , \xi^ \pi , x_0, X),
$$
where $\hat{X}= \cup_{tr \in \xi^a} X_{tr}$ is the set of states added during the replacement and $X^\pi_m=X$ is the set of marked states. Note that $G^\pi$ is a nondeterministic automaton, 
that is, $\xi^ \pi$ is a mapping $\xi^ \pi : X^\pi \times \Sigma \rightarrow 2^{X^\pi}$. The languages generated and marked by nondeterministic automaton are defined respectively as
\begin{align*}
	& L(G^\pi) = \{ w \in \Sigma ^*: \xi^ \pi (x_0,w) \not= \emptyset \} \\
	& L_m(G^\pi) = \{ w \in L(G): \xi^ \pi (x_0,w) \cap X_m \not= \emptyset \}.
\end{align*}

It is not difficult to see that
\begin{align*}
	& L(G^\pi) = \overline{\Theta ^\pi(L(G))} \\
	& L_m(G^\pi) = \Theta ^\pi(L(G)).
\end{align*}

\emph{Step 2}. We replace unobservable transitions by $\varepsilon$-transitions and denote the resulting automaton as
$$
G_{\varepsilon}^\pi=( X \cup \hat{X}, \Sigma_o , \xi _\varepsilon^\pi , x_0, X)
$$
where $\xi _\varepsilon^\pi = \{ (x, \sigma, x'): (x, \sigma, x') \in \xi ^\pi \wedge \sigma \in \Sigma _o \} \cup \{ (x, \varepsilon, x'): (x, \sigma, x') \in \xi ^\pi \wedge \sigma \not\in \Sigma _o \} $. 

\emph{Step 3}. We convert the nondeterministic automaton $G_{\varepsilon} ^\pi$ to a deterministic automaton $G_{obs}^\pi$, called CA-observer, as follows.
\begin{align*}
	G_{obs}^\pi & =(Y,\Sigma _o, \zeta, y_0, Y_m) \\
	&= Ac(2^{X \cup \hat{X}},\Sigma _o, \zeta, UR(\{ x_0\}), Y_m), 
\end{align*}
where $Ac(.)$ denotes the accessible part; $UR(.)$ is the unobservable reach defined, for $y \subseteq X \cup \hat{X}$, as
$$
UR(y) = \{ x \in X \cup \hat{X}: (\exists x' \in y) x\in \xi _{\varepsilon}^\pi (x', \varepsilon) \}.
$$

The transition function $\zeta$ is defined, for $y \in Y$ and $\sigma \in \Sigma _o$ as
$$
\zeta (y, \sigma) = UR(\{x \in X \cup \hat{X}: (\exists x' \in y)x\in \xi _\varepsilon^\pi (x',\sigma) \}).
$$
The marked states are defined as
$$
Y_m = \{y \in Y: y \cap X \not = \emptyset \}.
$$

It is not difficult to see that
\begin{align*}
	& L(G_{obs}^\pi) = \overline{\Phi ^\pi(L(G))} \\
	& L_m(G_{obs}^\pi) = \Phi ^\pi(L(G)).
\end{align*}

From Theorem \ref{Theorem5}, we have the following corollary.

\begin{corollary} \label{Corollary1} 
	
	$J$ is CA-D-observable with respect to $L(G)$, $\Sigma _{o}$, $\Sigma ^a_o$ and $\Phi^\pi$ if and only if for all $y \in Y_m$ and $\sigma \in \Sigma$,
	\begin{align*}
		& (\forall x \in y) x \in X_H \wedge \xi (x,\sigma) \in X \Rightarrow \xi (x,\sigma) \in X_H \\
		\Leftrightarrow \
		& (\exists x' \in y) x' \in X_H \wedge \xi (x',\sigma) \in X_H .
	\end{align*}
	
\end{corollary}

{Since the number of states in $G_{obs}^\pi$ is bounded by $2^{|X \cup \hat{X}|}$, the computational complexity of checking CA-D-observability using Corollary \ref{Corollary1} is $O(|X| |\Sigma_o| 2^{|X \cup \hat{X}|})$. It is possible that more efficient methods exist to check CA-D-observability. The computational complexity of checking CA-D-controllability can be obtained as follows. (1) The computational complexity of checking $J (\Sigma _{uc} \cup \Sigma ^a_c) \cap L(G) \subseteq J$ is the same as the computational complexity of checking conventional controllability \cite{cassandras2021introduction}, which is $O(|X| |\Sigma|)$. (2) The computational complexity of checking $J \subseteq (\Sigma - \Sigma^a_c)^*$ is also $O(|X| |\Sigma|)$.
}

\section{An Illustrative Example} \label{s6}

To illustrate deterministic and nonblocking supervisory control under cyber-attacks, consider a robot moving in a space. The space is divided into $5 \times 5 = 25$ cells as shown in Fig. \ref{fig:Figure2}. 	

\begin{figure}[htb!]
	\centering
	\includegraphics[scale=0.4]{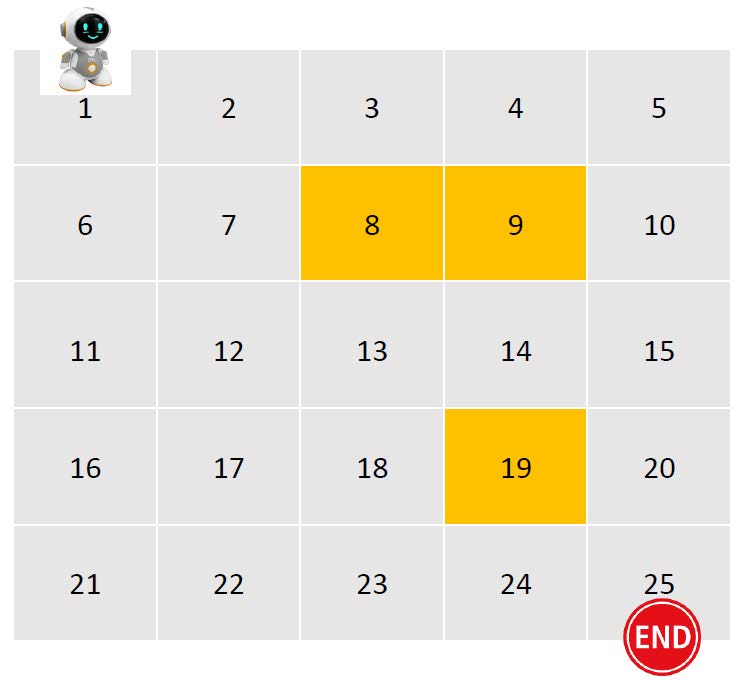}
	\caption{A robot moves in a space. }
	\label{fig:Figure2}
\end{figure}

From any cell, the robot can move either to the cell on the right in the figure or to the cell below (but not to the cell on the left or to the cell above). Initially, the robot is in Cell 1 and the goal is to move to Cell 25. Define the events as
\begin{align*}
	& e1, e2, e3, e4, e5: \mbox{the robot moves to the cell on the right} \\
	& d1, d2, d3, d4, d5: \mbox{the robot moves to the cell below} .
\end{align*}
The automaton 
\begin{align*}
	& G=( X, \Sigma , \xi , x_0, X_m ) \\
	= & ( \{1, ..., 25 \}, \{ei, di : i=1,...,5 \} , \xi , 1, \{25\} )
\end{align*} 
describing the movement of the robot is shown in Fig. \ref{fig:G}. Since, for deterministic supervisory control, marked states are irrelevant. So, for convenience, let us assume that all states are marked, that is, $X=X_m$ until we discuss nonblocking control later.

\begin{figure}[htb!]
	\centering
	\includegraphics[scale=0.5]{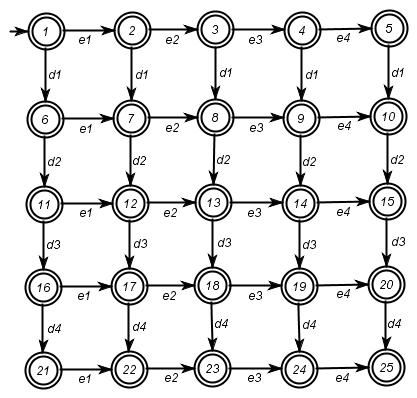}
	\caption{The automaton $G$ describing the movement of the robot. The initial state is $1$, denoted by $\rightarrow$. All states are marked, denoted by double circles.}
	\label{fig:G}
\end{figure}

Assume that there are 3 objects in Cells 8, 9, and 19, respectively. To avoid collision, the robot shall not enter these cells. In other words, States 8, 9, and 19 are illegal. Therefore, the automaton $H$ is shown in Fig. \ref{fig:H}.

\begin{figure}[htb!]
	\centering
	\includegraphics[scale=0.5]{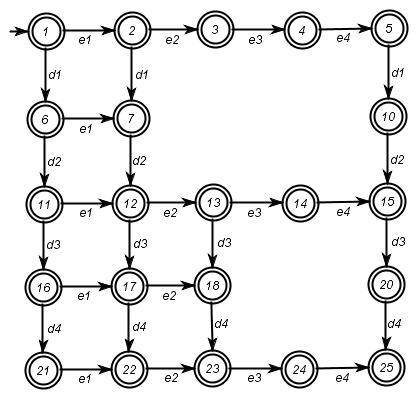}
	\caption{The automaton $H$ generating the specification language $J$. }
	\label{fig:H}
\end{figure}

To focus on cyber-attacks, we assume that all events are controllable and observable, that is, $\Sigma _c = \Sigma _o = \Sigma$. This means, in particular, $G^\pi =G_{\varepsilon}^\pi$. 

We first investigate sensor attacks and assume $\Sigma ^a_c = \emptyset$. Let us consider different cases and construct the corresponding ALTER model are as follow.

\emph{Case 1.} Event $e2$ in all transitions $tr=(x, e2, x')$ is replaced with $e1$ by the attacker, that is, for all $tr=(x, e2, x')$, $A_{tr}=\{e1\}$ (replacement attack). 
In this case, $G^\pi=G_{\varepsilon}^\pi$ can be obtained by replacing $e2$ with $e1$ as shown in Fig. \ref{fig:Ga1} without introducing $\varepsilon$ transitions. 

\begin{figure}[htb!]
	\centering
	\includegraphics[scale=0.5]{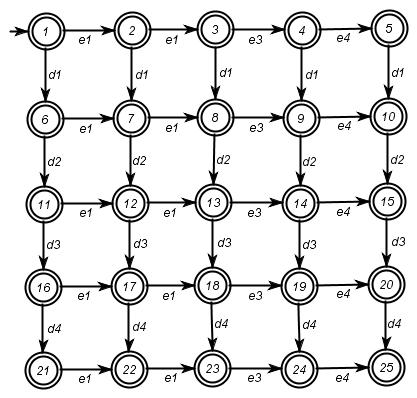}
	\caption{The automaton $G^\pi=G_{\varepsilon}^\pi$ for Case 1. }
	\label{fig:Ga1}
\end{figure}

Since $G^\pi$ is a deterministic automaton, the CA-observer $G_{obs}^\pi$ is isomorphic to $G^\pi$. In this case, $J$ is CA-D-observable with respect to $L(G)$, $\Sigma _{o}$, $\Sigma ^a_o$ and $\Phi^\pi$. In fact, we can prove the following general result.

\begin{proposition} \label{Proposition4} 
	
	Consider a discrete event system $G$ under cyber-attacks and a nonempty closed language $J \subseteq L(G)$ generated by $H$. 
	If $G_{\varepsilon}^\pi$ is a deterministic automaton, then $J$ is CA-D-observable with respect to $L(G)$, $\Sigma _{o}$, $\Sigma ^a_o$ and $\Phi^\pi$.
	
\end{proposition}
\noindent {\em Proof}

Since $G_{\varepsilon}^\pi$ is a deterministic automaton, $G_{\varepsilon}^\pi$ and $G_{obs}^\pi =(Y,\Sigma _o, \zeta, y_0, Y_m)$ are isomorphic to $G^\pi$. 
Hence, all states $y \in Y$ are singletons, that is $y= \{ x \}$. Therefore, the condition in Corollary \ref{Corollary1} becomes
\begin{align*}
	& x \in X_H \wedge \xi (x,\sigma) \in X \Rightarrow \xi (x,\sigma) \in X_H \\
	\Leftrightarrow \
	& x \in X_H \wedge \xi (x,\sigma) \in X_H ,
\end{align*}
which is true, because $H$ is a sub-automaton of $G$.

\hfill \bull

\emph{Case 2.} Event $e1$ in all transitions $tr=(x, e1, x')$ is deleted by the attacker, that is, for all $tr=(x, e1, x')$, $A_{tr}=\{ \varepsilon \}$ (deletion attack). In this case, $G^\pi=G_{\varepsilon}^\pi$ is shown in Fig. \ref{fig:Ga2}. 

\begin{figure}[htb!]
	\centering
	\includegraphics[scale=0.5]{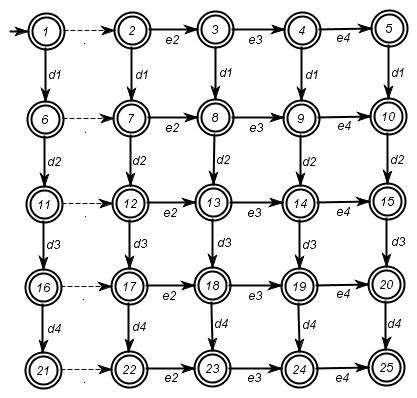}
	\caption{The automaton $G^\pi=G_{\varepsilon}^\pi$ for Case 2. }
	\label{fig:Ga2}
\end{figure}

The corresponding CA-observer $G_{obs}^\pi$ is shown in Fig. \ref{fig:observerGa2}. It can be checked that the condition in Corollary \ref{Corollary1} is satisfied. 
Therefore, $J$ is CA-D-observable with respect to $L(G)$, $\Sigma _{o}$, $\Sigma ^a_o$ and $\Phi^\pi$.

\begin{figure}[htb!]
	\centering
	\includegraphics[scale=0.5]{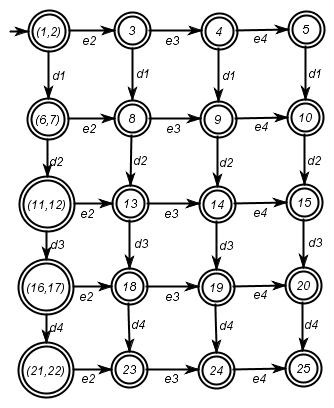}
	\caption{The CA-observer $G_{obs}^\pi$ for Case 2. }
	\label{fig:observerGa2}
\end{figure}

\emph{Case 3.} Event $d2$ in transitions $tr_1=(8, d2, 13)$ is replaced by $d2 \{d2, d3\}^*$, that is, $A_{tr_1}=d2 \{d2, d3\}^*$ and Event $e2$ in transitions $tr_2=(12, e2, 13)$ is replaced by $e2 \{d2, d3\}^*$, that is, $A_{tr_2}=e2 \{d2, d3\}^*$. In this case, $G^\pi=G_{\varepsilon}^\pi$ is shown in Fig. \ref{fig:Ga3}. 

\begin{figure}[htb!]
	\centering
	\includegraphics[scale=0.5]{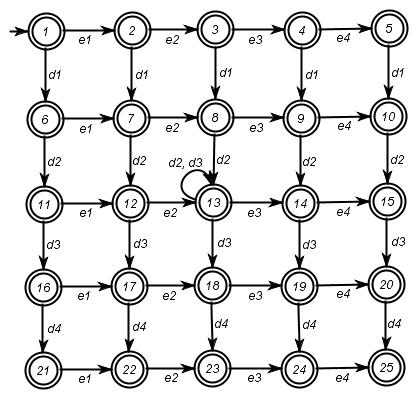}
	\caption{The automaton $G^\pi=G_{\varepsilon}^\pi$ for Case 3, }
	\label{fig:Ga3}
\end{figure}

The corresponding CA-observer $G_{obs}^\pi$ is shown in Fig. \ref{fig:observerGa3}. Let us check the condition in Corollary \ref{Corollary1}. For $y=(13, 18)$,
\begin{align*}
	& 18 \in X_H \wedge \xi (18,e3) = 19 \in X \wedge \xi (18,e3) = 19 \not\in X_H \\
	\Rightarrow \
	& \neg (\forall x \in y) x \in X_H \wedge \xi (x,\sigma) \in X \Rightarrow \xi (x,\sigma) \in X_H ,
\end{align*}
and 
\begin{align*}
	& 13 \in X_H \wedge \xi (13,e1) = 14 \in X_H \\
	\Rightarrow \
	& (\exists x' \in y) x' \in X_H \wedge \xi (x',\sigma) \in X_H .
\end{align*}
Hence, the condition in Corollary \ref{Corollary1} is not satisfied. Therefore, $J$ is not CA-D-observable with respect to $L(G)$, $\Sigma _{o}$, $\Sigma ^a_o$ and $\Phi^\pi$.

\begin{figure}[htb!]
	\centering
	\includegraphics[scale=0.5]{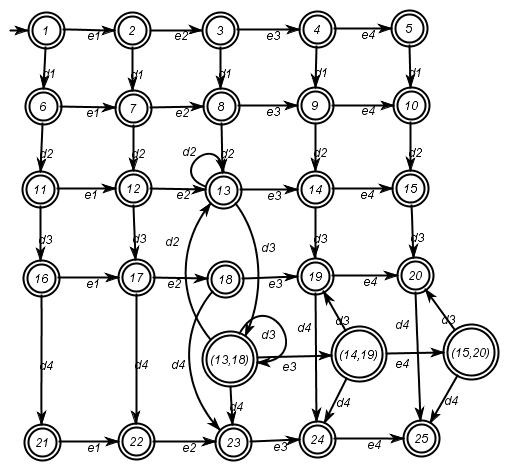}
	\caption{The CA-observer $G_{obs}^\pi$ for Case 3.}
	\label{fig:observerGa3}
\end{figure}

We now investigate actuator attacks and assume $\Sigma ^a_o = \emptyset$. The first condition in CA-D-controllability, $J (\Sigma _{uc} \cup \Sigma ^a_c) \cap L(G) \subseteq J$, is satisfied if and only if $d1, d3, e2, e3 \not\in \Sigma^a_c$. 
The second condition in CA-D-controllability, $J \subseteq (\Sigma - \Sigma^a_c)^*$, is satisfied if and only if $\Sigma^a_c = \emptyset$. 
Therefore, to ensure CA-D-controllability is satisfied, it is required that no controllable events are attackable. This is rather restrictive, but it is necessary for the existence of a deterministic supervisor. 
Intuitively, this is because if any controllable event is attackable, then the attacker can disable it and hence prevent the supervised system from achieving $J$.

If $J$ is both CA-D-controllable with respect to $L(G)$, $\Sigma _{uc}$, and $\Sigma ^a_c$, and CA-D-observable with respect to $L(G)$, $\Sigma _{o}$, $\Sigma ^a_o$ and $\Phi^\pi$, as in Case 2, 
then a deterministic supervisor $\cS_p$ (= $\cS_r$) can be designed based on the CA-observer $G_{obs}^\pi$ in Fig. \ref{fig:observerGa2}. The disablement by supervisor $\cS_{p}$ is shown in Table \ref{Table1}.

\begin{table}[htb!]
	\label{Table1}
	\centering
	\caption{Disablement by supervisor $\cS_{p}$}
	\begin{tabular}{|c|c|}
		\hline
		state  &  disablement    \\
		\hline 
		3   & $d1$ \\
		\hline 
		4   & $d1$ \\
		\hline 
		(6,7)   & $e2$ \\
		\hline 
		14   & $d3$ \\
		\hline
		18   & $e3$ \\
		\hline
	\end{tabular}
\end{table}

Next, we illustrate nonblocking supervisory control under cyber-attacks. Assume that the marked state is $X_m = \{ 25 \}$. The new $G$ and $H$ are shown in Fig. \ref{fig:Gm} and Fig. \ref{fig:Hm}, respectively. Since
$$
X_{m,H} = X_H \cap X_m = \{25\}, 
$$
by Theorem \ref{Theorem6}, $K=L_m(H)$ is $L_m(G)$-closed. Therefore, nonblocking supervisory control exists if and only if deterministic supervisory control exists. In particular, nonblocking supervisory control exists for Case 1 and Case 2, but not for Case 3.

\begin{figure}[htb!]
	\centering
	\includegraphics[scale=0.5]{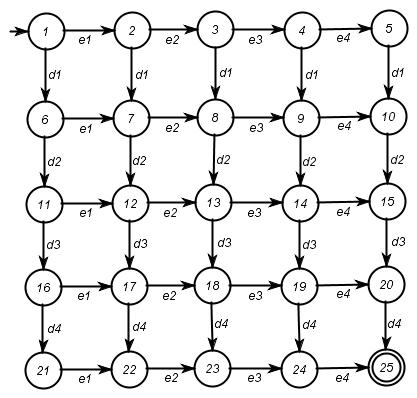}
	\caption{The new automaton $G$ for nonblocking supervisory control.}
	\label{fig:Gm}
\end{figure}

\begin{figure}[htb!]
	\centering
	\includegraphics[scale=0.5]{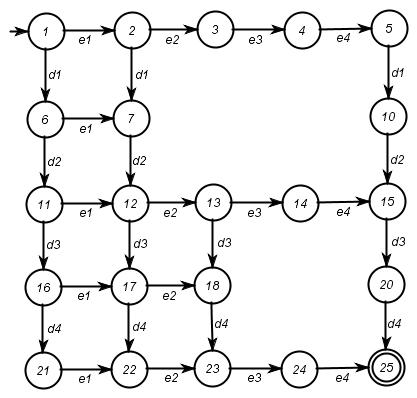}
	\caption{The new automaton $H$ for nonblocking supervisory control.}
	\label{fig:Hm}
\end{figure}

\section{Conclusions} \label{s7}

{We investigate supervisory control of discrete event systems under cyber-attacks to ensure a deterministic (unique) closed-loop language and reachability to marked states. The main contributions of the paper are as follows.}
(1) Deterministic supervisory control problem and nonblocking supervisory control problem are proposed to meet the needs of control of cyber-physical systems under cyber-attacks. 
(2) CA-D-controllability and CA-D-observability are introduced and shown to be the existence condition of deterministic supervisory control. 
(3) The relation between deterministic supervisory control and nonblocking supervisory control under cyber-attacks is discovered and existence condition for nonblocking supervisory control is derived. 
(4) A method for checking CA-D-controllability and CA-D-observability is developed and implemented using automata. In future work, we will consider what to do if CA-D-controllability and CA-D-observability are not satisfied.

%


\begin{biography}[{\includegraphics[width=1in,height=1.25in,clip,keepaspectratio]{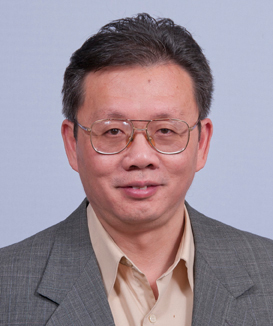}}]{Feng Lin}
	(S'85-M'88-SM'07-F'09) received his B.Eng. degree in electrical engineering from Shanghai Jiao Tong University, Shanghai, China, in 1982, and the M.A.Sc. and Ph.D. degrees in electrical engineering from the University of Toronto, Toronto, ON, Canada, in 1984 and 1988, respectively. He was a Post-Doctoral Fellow with Harvard University, Cambridge, MA, USA, from 1987 to 1988. Since 1988, he has been with the Department of Electrical and Computer Engineering, Wayne State University, Detroit, MI, USA, where he is currently a Professor. 
	
	His current research interests include discrete event systems, hybrid systems, neural networks, robust control, and their applications in alternative energy, biomedical systems, machine learning, and automotive control. He authored a book entitled ``Robust Control Design: An Optimal Control Approach'' and coauthored a paper that received a George Axelby outstanding paper award from the IEEE Control Systems Society. 
	
\end{biography}

\begin{biography}[{\includegraphics[width=1in,height=1.25in,clip,keepaspectratio]{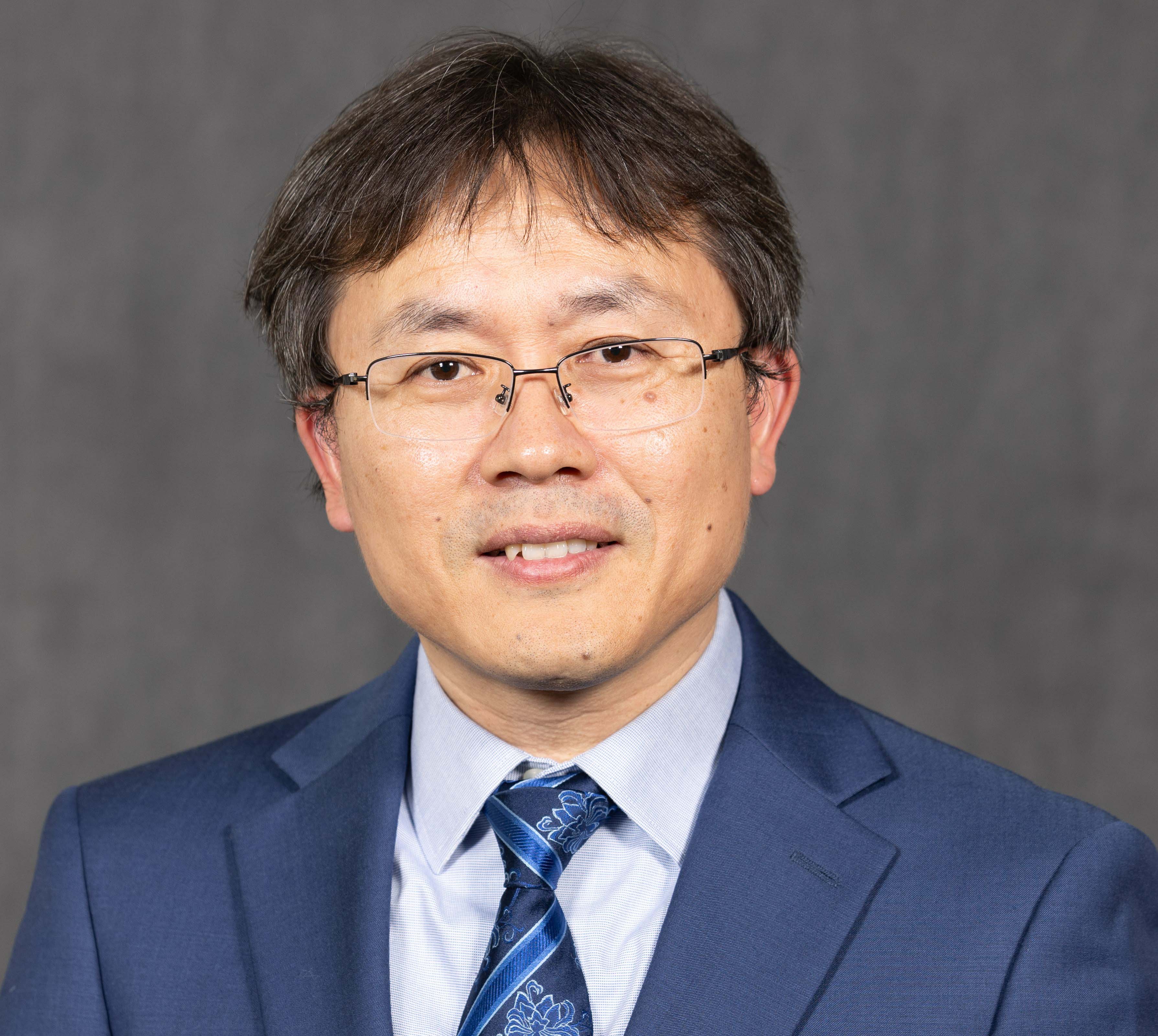}}]{Caisheng Wang}
	(F'25) received the BS and MS degrees from Chongqing University, China, in 1994 and 1997, respectively, and the Ph.D. degree from Montana State University, Bozeman, MT, in 2006, all in electrical engineering. 
	
	From August 1997 to May 2002, he worked as an electrical engineer and then as a vice department chair at Zhejiang Electric Power Test \& Research Institute, Hangzhou, China. Since August 2006, he has joined Wayne State University, where he is currently a Professor in the Department of Electrical and Computer Engineering. His current research interests include modeling and control of power systems and electric vehicles, power electronics, energy storage devices, distributed generation and Microgrids, alternative/hybrid energy power generation systems, and fault diagnosis and on-line monitoring of electric apparatus. He is an Associate Editor of IEEE Electrification Magazine and was an Associate Editor of several other journals, including IEEE Transactions on Smart Grid.	
	
\end{biography}

\begin{biography}[{\includegraphics[width=1in,height=1.25in,clip,keepaspectratio]{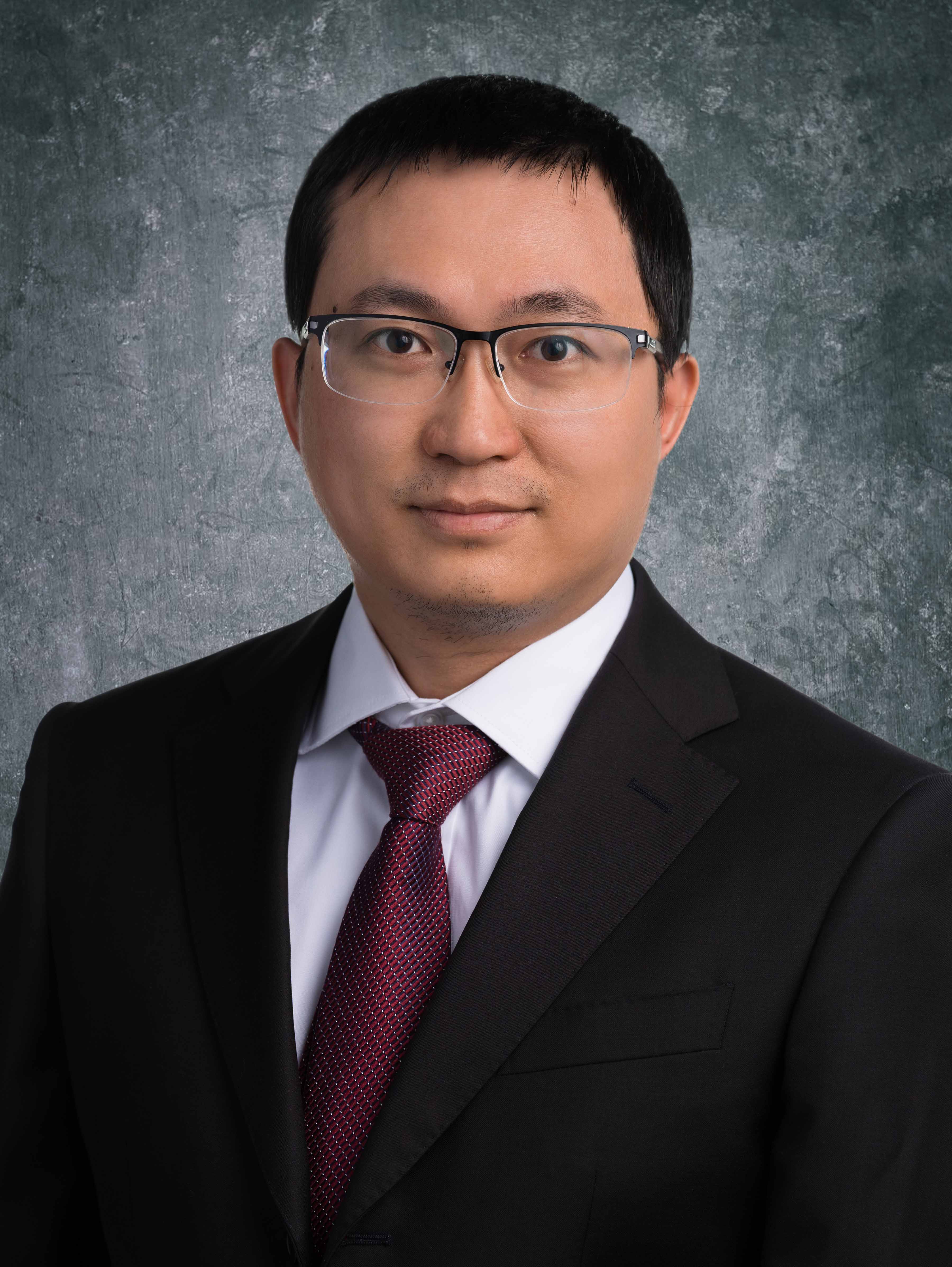}}]{Jun Chen}
	(S'11-M'14-SM'20) received his Bachelor's degree in Automation from Zhejiang University, Hangzhou China, in 2009, and Ph.D. in Electrical Engineering from Iowa State University, Ames IA, USA, in 2014. He was with Idaho National Laboratory from 2014 to 2016 and with General Motors from 2017 to 2020. Dr. Chen joined Oakland University in 2020, where he is currently an associate professor at the ECE department.	
	
	His research interests include advanced control and optimization, model predictive control, artificial intelligence, and stochastic hybrid systems, with applications in intelligent vehicles, robotics, and energy systems. Dr. Chen is a recipient of the NSF CAREER Award, the Best Paper Award from \textsc{IEEE Transactions on Automation Science and Engineering}, and the Best Paper Award from \textsc{IEEE International Conference on Electro Information Technology}, the Best Paper Award from \textsc{IEEE Cyber Awareness \& Research Symposium}, the New Investigator Research Excellence Award and Outstanding Graduate Mentor Award from Oakland University, the Publication Achievement Award from Idaho National Laboratory, the Research Excellence Award from Iowa State University, and the Outstanding Student Award from Zhejiang University. He is currently a Senior Member of the IEEE. 
	
\end{biography}

\begin{biography}[{\includegraphics[width=1in,height=1.25in,clip,keepaspectratio]{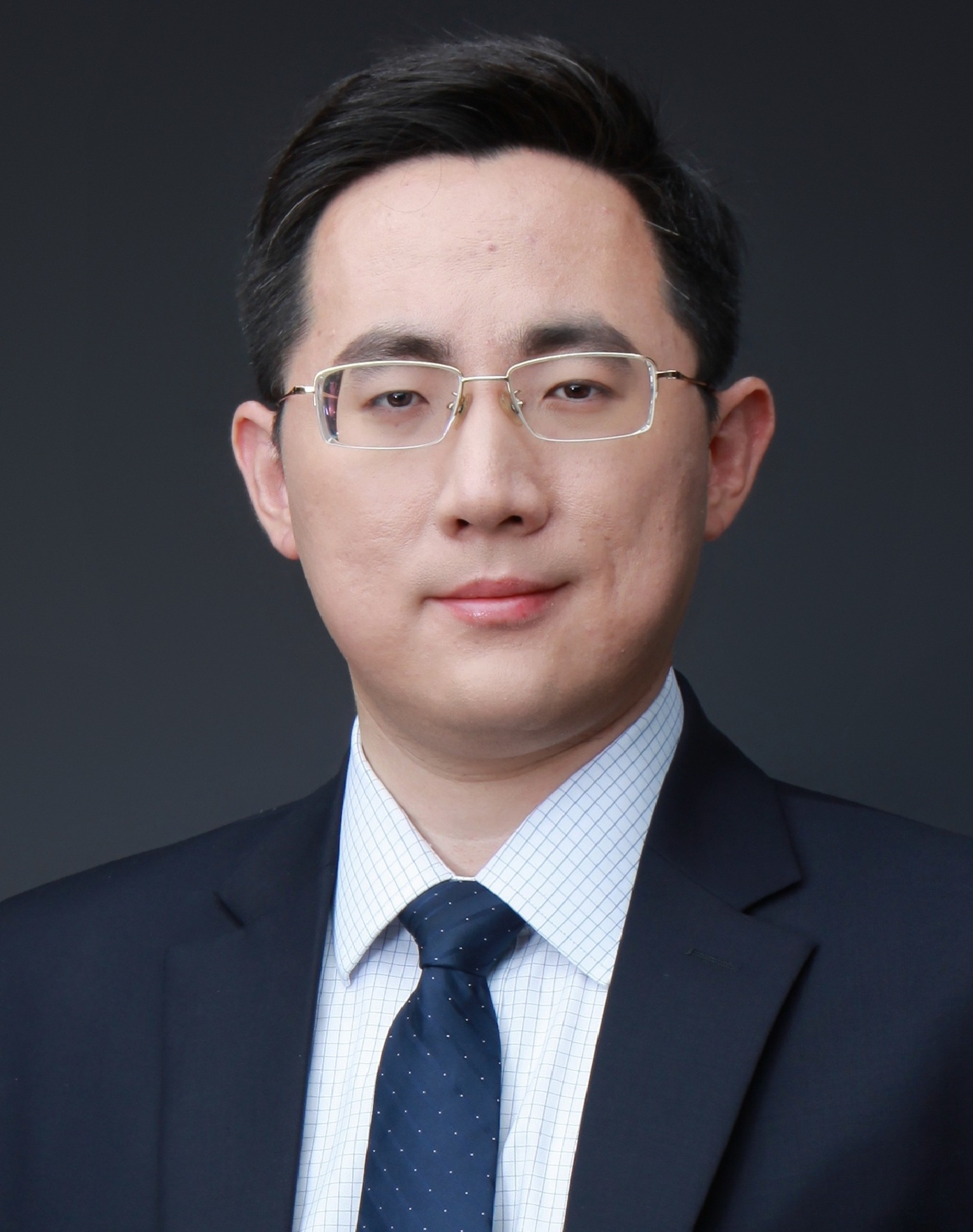}}]{Xiang Yin}
	(S'14-M'17)  was born in Anhui, China, in 1991. He received the B.Eng degree from Zhejiang University in 2012, the M.S. degree from the University of Michigan, Ann Arbor, in 2013, and the Ph.D degree from the University of Michigan, Ann Arbor, in 2017, all in electrical engineering. Since 2017, he has been with the School of Automation and Intelligent Sensing, Shanghai Jiao Tong University, where he is a Full Professor. His research interests include formal methods, discrete-event systems, robotics, artificial intelligence and cyber-physical systems. 
	
	Dr. Yin is serving as the chair of the \emph{IEEE CSS Technical Committee on Discrete Event Systems},  Associate Editors for the \emph{Journal of Discrete Event Dynamic Systems: Theory \& Applications}, \emph{Nonlinear Analysis: Hybrid Systems}, \emph{IEEE Control Systems Letters}, \emph{IEEE Transactions on Automation Science and Engineering},  and a member of the \emph{IEEE CSS Conference Editorial Board}. 
	
\end{biography}

\end{document}